\documentclass[journal]{IEEEtran}

\usepackage{amsmath,amstext,amsfonts,amssymb,eucal,graphicx}
\usepackage{subfigure}
\usepackage{epsfig}
\usepackage{listings}

\newtheorem{lemma}{Lemma}
\newtheorem{corollary}{Corollary}

\newtheorem{theorem}{Theorem}
\newtheorem{definition}{Definition}

\newcommand{\E}{\mathbb{E}}

\begin{document}
\bibliographystyle{IEEEtran}%IEEEtran
\title{Convolution operations arising from Vandermonde matrices}
\author{
  \IEEEauthorblockN{\O yvind Ryan,~\IEEEmembership{Member,~IEEE} and M{\'e}rouane~Debbah,~\IEEEmembership{Senior Member,~IEEE}\\}
  \thanks{This work was supported by Alcatel-Lucent within the Alcatel-Lucent Chair on flexible radio at SUPELEC}
  \thanks{\O yvind~Ryan is with the Centre of Mathematics for Applications, University of Oslo, P.O. Box 1053 Blindern, NO-0316 Oslo, Norway, oyvindry@ifi.uio.no}
  \thanks{M{\'e}rouane~Debbah is with SUPELEC, Gif-sur-Yvette, France, merouane.debbah@supelec.fr}
}

\markboth{IEEE Transactions on Information Theory}{Shell \MakeLowercase{\textit{et al.}}: Bare
Demo of IEEEtran.cls for Journals}

\maketitle
\begin{abstract}
Different types of convolution operations involving large Vandermonde matrices are considered.
The convolutions parallel those of large Gaussian matrices and additive and multiplicative free convolution.
First additive and multiplicative convolution of Vandermonde matrices and deterministic diagonal matrices are considered.
After this, several cases of additive and multiplicative convolution of two independent Vandermonde matrices are considered.
It is also shown that the convergence of any combination of Vandermonde matrices is almost sure.
We will divide the considered convolutions into two types: those which depend on the phase distribution of the Vandermonde matrices,
and those which depend only on the spectra of the matrices. A general criterion is presented to find which type applies for any given convolution.
A simulation is presented, verifying the results.
Implementations of all considered convolutions are provided and discussed,
together with the challenges in making these implementations efficient.
The implementation is based on the technique of Fourier-Motzkin elimination, and is quite general as it can be applied to virtually any combination of Vandermonde matrices.
Generalizations to related random matrices, such as Toeplitz and Hankel matrices, are also discussed.
\end{abstract}


\section{Introduction}
Certain random matrices have in the large dimensional limit a deterministic behavior of the eigenvalue distributions,
meaning that one can compute the eigenvalue distributions of ${\bf A}{\bf B}$ and ${\bf A}+{\bf B}$ based only on the individual eigenvalue distributions
of ${\bf A}$ and ${\bf B}$, when the matrices are independent and large. 
The process of computing theses eigenvalues is called {\bf convolution}, 
or {\bf de-convolution} when one would like to compute the inverse operation. 
Gaussian-like matrices fit into this setting, and the concept which can be used to find the eigenvalue distribution 
from that of the component matrices in this case is called freeness~\cite{book:hiaipetz}.
Free probability
theory~\cite{book:hiaipetz}, which uses the concept of freeness, is not a new tool but has grown into an
entire field of research since the pioneering work of Voiculescu in the 1980's~(\cite{vo2,paper:vomult,vo6,vo7}). However, the basic
definitions of free probability are quite abstract and this has hinged a burden on its actual practical use. 
The original goal was to introduce an analogy to independence in classical probability that can be used for non-commutative random variables like matrices. 
These more general random variables are elements of what is  called a  {\em noncommutative probability space}. 
The convolution/deconvolution techniques used are various. The classical ones are either analytic (using $R$ and $S$ transforms \cite{Berco+Vovo.93,book:hiaipetz}) 
or based on moments~\cite{book:comblect,VDN,fbg.inf.div.rect,Vo.104}.
Recent deconvolution techniques based on statistical eigen-inference
methods using large Wishart matrices~\cite{raomingospeicheredelman}, random
Matrix theory~\cite{elkaroui2} or other  deterministic
equivalents {\it \`a la} Girko~\cite{book.girko98,book:girkogest,paper:hachem07,paper:mestresub}
were proposed and are possible alternatives. Each one has its advantages and
drawbacks.  Unfortunately, although successfully applied \cite{eurecom:freedeconvinftheory,eurecom:channelcapacity},
all these techniques can only treat very simple models i.e. the case where
one of the considered matrices is unitarily invariant. This
invariance has a special meaning in wireless networks and supposes
that there is some kind of symmetry in the problem to be analyzed.  The moments technique, which will be the focus of this work, is very appealing and
powerful in order to derive the exact asymptotic moments of "non-free matrices", for which we still do not have a general framework. It
requires combinatorial skills and can be used for a large class of
random matrices. The main drawback of the technique (compared to other tools such as
the Stieltjes transform method~\cite{paper:doziersilverstein1}) is
that it can rarely provide the exact eigenvalue distribution.
However, in many  applications, one needs only a subset of
the moments depending on the number of parameters to be estimated.

Recently~\cite{ryandebbah:vandermonde1}, Vandermonde matrices (which do not fall within the free probability framework) 
were  shown to be a case of high interest in wireless communications. 
Such matrices have various applications in signal reconstruction \cite{supelec:estimation}, 
cognitive radio \cite{paper:sampaiokobayashi2}, physical layer security \cite{Kobayashidebbah}, 
and MIMO channel modeling \cite{newfrommerouane2}. A Vandermonde matrix with entries on the unit circle is on the form
\begin{equation} \label{vandermonde}
  {\bf V} = \frac{1}{\sqrt{N}}
               \left( \begin{array}{lll}
                        1                    & \cdots & 1 \\
                        e^{-j \omega_1}      & \cdots & e^{-j \omega_L} \\
                        \vdots               & \ddots & \vdots \\
                        e^{-j (N-1)\omega_1} & \cdots & e^{-j (N-1)\omega_L}
                      \end{array}
               \right)
\end{equation}
${\bf V}$ will in this paper always denote a Vandermonde matrix, and its dimension will be denoted $N\times L$.
The $\omega_1$,...,$\omega_L$, also called phase distributions, will be assumed i.i.d., taking values in $[ 0,2\pi )$.
We will also assume, as in many applications, that $N$ and $L$ go to infinity at the same rate, and write $c=\lim_{N\rightarrow\infty}\frac{L}{N}$ for the aspect ratio.
If necessary, we will write ${\bf V}_{\omega}$ to emphasize the actual phase distribution,
or ${\bf V}_{\omega,c}$ to also emphasize the aspect ratio.
In~\cite{ryandebbah:vandermonde1}, the limit eigenvalue distributions of combinations of ${\bf V}^H{\bf V}$
and diagonal matrices ${\bf D}(N)$ were shown to be dependent on only the
limit eigenvalue distributions of the two matrices. Important combinations are the multiplicative and additive models,
\begin{equation} \label{firsttypes}
  {\bf D}(N){\bf V}^H{\bf V} \mbox{ and } {\bf D}(N) + {\bf V}^H{\bf V}.
\end{equation}
In the large $N$-limit, (\ref{firsttypes}) thus gives rise to two convolution operations,
\begin{description}
  \item[1)] $\lim_{N\rightarrow\infty} {\bf D}(N){\bf V}^H{\bf V}$ and $\lim_{N\rightarrow\infty} ({\bf D}(N)+{\bf V}^H{\bf V})$,
\end{description}
which thus depend only on the input spectra.
Here $\lim$ is used to denote the limit of the eigenvalue distribution of the considered matrix, in an appropriate metric.
However, it is not clear from~\cite{ryandebbah:vandermonde1} how 1) can be computed algorithmically, as only sketches for this were provided.
We also have the operations
\begin{description}
  \item[2)] $\lim_{N\rightarrow\infty} {\bf D}(N){\bf V}{\bf V}^H$ and $\lim_{N\rightarrow\infty} ({\bf D}(N) + {\bf V}{\bf V}^H)$,
\end{description}
for which it is unknown whether the result only depends on the spectra. 
This case happens in practical scenarios (for cognitive applications~\cite{paper:sampaiokobayashi2} as well as secure transmissions~\cite{Kobayashidebbah}) 
when a Vandermonde precoder ${\bf V}$ is used in a given Toeplitz channel matrix ${\bf D}(N)$ independent from ${\bf V}$. 
One can then compute cognitive and secrecy rates. When we replace with independent Vandermonde matrices ${\bf V}_1$ and ${\bf V}_2$ 
which may or may not have the same phase distributions,
it is also unknown if the convolution operations
\begin{description}
  \item[3)] $\lim_{N\rightarrow\infty} {\bf V}_1^H{\bf V}_1{\bf V}_2^H{\bf V}_2$ and $\lim_{N\rightarrow\infty} ({\bf V}_1^H{\bf V}_1 + {\bf V}_2^H{\bf V}_2)$,
  \item[4)] $\lim_{N\rightarrow\infty} {\bf V}_1{\bf V}_1^H{\bf V}_2{\bf V}_2^H$ and $\lim_{N\rightarrow\infty} ({\bf V}_1{\bf V}_1^H + {\bf V}_2{\bf V}_2^H)$,
\end{description}
only depend on the spectra of ${\bf V}_1$ and ${\bf V}_2$. These cases are important  for the recovery of the distribution of sensors (which are deployed in a clustered manner with different mean positions) and in the case of MIMO multi-fold scattering \cite{mullermodelcomm}.

Expressions such as 4), when different types of matrices are multiplied,
will in the following be called {\em mixed moments}.

In this contribution we explain which of the above operations depend only on the spectra of the matrices,
state expressions for those convolutions (in fact, we also state expressions for the cases where the result can not be written in terms of the spectra),
explain how these expressions have been obtained algorithmically,
and explain an accompanying software implementation~\cite{eurecom:vandermondeimpl2,rmtdoc} of the corresponding algorithms.
We also attempt to complete the analysis started in~\cite{ryandebbah:vandermonde1},
by stating a very general criterion for when the mixed moments of (many) Vandermonde matrices and deterministic matrices depend only the input spectra:

{\em If there are no terms on the form ${\bf V}_1^H{\bf V}_2$ in a mixed moment,
with ${\bf V}_1$ and ${\bf V}_2$ independent and with different phase distributions,
the mixed moment will depend only on the spectra of the input matrices.
In all other cases, we can't expect dependence on just the spectra of the input matrices,
and the mixed moment can depend on the entire phase distributions of the input matrices.}

The software implementation can in fact be extended to handle all cases which meet this criterion,
as well as cases where knowledge of the phase distribution also is required.
In this way it is an indispensable tool, as it automates the very tedious computations inherent in the presented formulas,
for which no simple expressions are known.

Concluding from the criterion, 1) will depend only on the spectra (as shown in~\cite{ryandebbah:vandermonde1}), as does 3).
4) may not depend on only the spectra when the two phase distributions are different.
Despite this, 4) is interesting in its own right, since it has a geometric interpretation in terms of phase distributions,
and is therefore handled separately.
For case 2), we state more generally that when the pattern ${\bf D}(N){\bf V}$ appears in a mixed moment, we can not expect dependence only on the spectrum.

It turns out that other types of random matrices can use the same methods as for Vandermonde matrices to compute their moments,
such as Toeplitz matrices and Hankel matrices. We will explain how the software implementation has been extended to handled these matrices as well.

The paper is organized as follows.
Section~\ref{section:essentials} provides background essentials on
random matrix theory needed for the main results, which are stated in~\ref{section:theorems}.
The results include the precise statement of the criterion above for when we only have dependence on the spectra of the matrices,
results on the convolution operations 1)-4),
and extensions to related random matrices such as Toeplitz and Hankel matrices.
A generalization of our results to almost sure convergence of matrices is also made.
All presented formulas are obtained from the implementation, and the major pieces in this implementation are gone through in Section~\ref{theoremimpl},
such as partition iteration, and Fourier-Motzkin elimination~\cite{paper:dahl1}.
Section~\ref{simulations} presents a simulation which verifies the results.

\section{Random matrix Background Essentials} \label{section:essentials}
In the following, upper (lower boldface) symbols will be used for
matrices (column vectors), whereas lower symbols will represent
scalar values, $(.)^T$ will denote the transpose operator, $(.)^\star$
conjugation, and $(.)^H=\left((.)^T\right)^\star$ hermitian
transpose. ${\bf I}_L$ will represent the $L\times L$ identity matrix.
We let $\mathrm{Tr}$ be the (non-normalized) trace for square matrices, defined by,
\[
  \mathrm{Tr}({\bf A}) = \sum_{i=1}^L a_{ii},
\]
where $a_{ii}$ are the diagonal elements of the $L\times L$ matrix ${\bf A}$.
We also let $\mathrm{tr}$ be the normalized trace, defined by $\mathrm{tr}({\bf A}) = \frac{1}{L}\mathrm{Tr}({\bf A})$.

In the following we will implicitly assume that $L$ and $N$ go to infinity in such a way that $\frac{L}{N}\rightarrow c$.
${\bf D}_r(N), 1\leq r\leq n$ will denote non-random diagonal $L\times L$ matrices.
We will have use for the following definition:
\begin{definition} \label{ddef}
  We will say that the $\{ {\bf D}_r(N) \}_{1\leq r\leq n}$ have a joint limit distribution as $N\rightarrow\infty$ if the limit
  \begin{equation} \label{alphadef}
    D_{i_1,...,i_s} = \lim_{N\rightarrow\infty} \mathrm{tr}\left( {\bf D}_{i_1}(N)\cdots {\bf D}_{i_s}(N)\right)
  \end{equation}
  exists for all choices of $i_1,...,i_s\in \{ 1,..,n\}$.
\end{definition}

A joint limit distribution for the ${\bf D}_r(N)$ will always be assumed in the following.
The corresponding concept for random matrices is the following:

\begin{definition}
Let $\{{\bf A}_n\}_{n=1}^{\infty}$ be an ensemble of (square) random matrices.
We say that $\{{\bf A}_n\}_{n=1}^{\infty}$ converge in distribution if the limit
\begin{equation} \label{moment}
  \lim_{n\rightarrow\infty} \E[\mathrm{tr}(({\bf A}_n)^r)]
\end{equation}
exists for all $r$.
We will say that ensembles $\{ {\bf A}_{1n},{\bf A}_{2n},...\}_{n=1}^{\infty}$ of random matrices converge in distribution if the limit
\begin{equation} \label{mixedmoment}
  \lim_{n\rightarrow\infty} \E[\mathrm{tr}({\bf A}_{i_1n}{\bf A}_{i_2n}\cdots {\bf A}_{i_sn})]
\end{equation}
exists whenever the matrix product ${\bf A}_{i_1n}{\bf A}_{i_2n}\cdots {\bf A}_{i_sn}$ is well-defined, and square.
\end{definition}

When we refer to moments, we will generally mean (\ref{moment}), while mixed moments refer to (\ref{mixedmoment}).
A stronger form of convergence, which we will generalize our results to, is {\em almost sure convergence in distribution}.
This type of convergence requires that (\ref{moment}), (\ref{mixedmoment}) are replaced with
\begin{eqnarray*}
  & & \mathrm{tr}\left(\left({\bf A}_n\right)^r\right) \stackrel{\mbox{a.s.}}{\rightarrow} C_r \\
  & & \mathrm{tr}({\bf A}_{i_1n}{\bf A}_{i_2n}\cdots {\bf A}_{i_sn}) \stackrel{\mbox{a.s.}}{\rightarrow} C_{i_1,...,i_s},
\end{eqnarray*}
where $C_r,C_{i_1,...,i_s}$ are constants.

We will also need some basic concepts from partition theory.
${\cal P}(n)$ will denote the partitions of $\{1,...,n\}$.
For a partition $\rho=\{ W_1,...,W_r\}\in{\cal P}(n)$, $W_1,...,W_r$ denote its blocks, while $|\rho|=r$ denotes the number of blocks,
$\|\rho\|=n$ the number of elements in the partition.
We will write $k\sim_{\rho}l$ when $k$ and $l$ belong to the same block of $\rho$.
We will also write $b(i)$ for the index of the block in $\rho$ $i$ belongs to.
Partition notation is adapted to the mixed moment (\ref{alphadef}) in the following way:
\begin{definition} \label{ddef2}
  For $\rho = \{ W_1,...,W_k \}$, with $W_i = \{ w_{i1},...,w_{i|W_i|} \}$,
  we define
  \begin{eqnarray}
    D_{W_i}  &=& D_{i_{w_{i1}},...,i_{w_{i|W_i|}}} \label{dblockdef} \\
    D_{\rho} &=& \prod_{i=1}^k D_{W_i}. \label{dpartdef}
  \end{eqnarray}
\end{definition}

The set of partitions is a partially ordered set under the refinement order, i.e. $\rho_1\leq\rho_2$ whenever any block of $\rho_1$ is contained within a
block of $\rho_2$. By $\rho_1\vee\rho_2$ we will mean the smallest partition (w.r.t. the refinement order) which is larger than both $\rho_1$ and $\rho_2$.
$\vee$ will in our results be used in conjunction with the partition $[0,1]_n\in{\cal P}(2n)$, defined by
\[ [0,1]_n=\{\{1,2\},\{3,4\},...,\{2n-1,2n\}\}.\]
$[0,1]_n$ is an example of what is called an interval partition, meaning that each block consists solely of successive numbers.
We will also write $[\cdot,\cdot]$ for the intervals in an interval partition, so that we could also have written
\[ [0,1]_n=\{[1,2],[3,4],...,[2n-1,2n]\}.\]

We will in the following consider the trace of a general mixed moment of Vandermonde matrices and deterministic matrices,
the only requirement being that matrices and their adjoints appear in alternating order so that the resulting matrix is square:
\begin{equation} \label{moregeneral}
  \mathrm{tr}\left({\bf D}_1(N){\bf V}_{i_1}^H {\bf V}_{i_2} \cdots {\bf D}_n(N){\bf V}_{i_{2n-1}}^H{\bf V}_{i_{2n}}\right),
\end{equation}
where ${\bf V}_1,{\bf V}_2,...$ are assumed independent and with phase distributions $\omega_1,\omega_2,...$.
In particular, we assume that $N_{i_{2k}}=N_{i_{2k-1}}$ when the ${\bf V}_i$ are $N_i\times L$,
in order for the dimensions of the matrices in (\ref{moregeneral}) to match.
It turns out we can obtain the asymptotic behavior of (\ref{moregeneral}) for arbitrary continuous phase distributions $\omega_i$.
For (\ref{moregeneral}) we will let $\sigma$ be the partition in ${\cal P}(2n)$ defined by equality of the phase distributions, i.e.
$j\sim_{\sigma}k$ if and only if $\omega_{i_j}=\omega_{i_k}$ ($i_j$ and $i_k$ may or may not be different for this).
Similarly we will let $\sigma_1$ be the partition in ${\cal P}(2n)$ defined by dependence of the Vandermonde matrices, i.e.
$j\sim_{\sigma_1}k$ if and only if $i_j=i_k$.
Obviously, $\sigma_1\leq\sigma$.

\section{Statement of main results} \label{section:theorems}
The main result of the paper addresses moments on the form (\ref{moregeneral}), and goes as follows.

\begin{theorem} \label{deconv}
Let ${\bf V}_i$ be independent $N_i\times L$ Vandermonde matrices with aspect ratios $c_i=\lim_{N_i\rightarrow\infty} \frac{L}{N_i}$
and phase distributions $\omega_i$ with continuous densities on $[0,2\pi)$.
The mixed moment
\begin{equation} \label{moregeneral2}
  \lim_{N\rightarrow\infty} \mathrm{tr}\left( {\bf D}_1(N){\bf V}_{i_1}^H {\bf V}_{i_2} \cdots {\bf D}_n(N){\bf V}_{i_{2n-1}}^H{\bf V}_{i_{2n}} \right).
\end{equation}
always exists when ${\bf D}_i(N)$ have a joint limit distribution.
When $\sigma\geq [0,1]_n$ (i.e. there are no terms on the form ${\bf V}_i^H{\bf V}_j$,
with ${\bf V}_i$ and ${\bf V}_j$ independent and with different phase distributions), (\ref{moregeneral2}) depends only on the moments
\begin{eqnarray*}
  V_n^{(i)}       &=& \lim_{N\rightarrow\infty} E\left[ \mathrm{tr}\left( \left( {\bf V}_i^H{\bf V}_i \right)^n \right)\right] \\
  D_{i_1,...,i_s} &=& \lim_{N\rightarrow\infty} \mathrm{tr}\left( {\bf D}_{i_1}(N)\cdots {\bf D}_{i_s}(N)\right),
\end{eqnarray*}
the aspect ratios $c_i$, and $\sigma$, and assumes the form
\begin{equation} \label{mostgeneral}
  \sum_{s,r,i_t,j_t,k_t} a_{i_1,...,i_s,j_1,...,j_r,k_1,...,k_r} D_{i_1,...,i_s} \prod_{t=1}^r V_{j_t}^{(k_t)},
\end{equation}
where the $a_{i_1,...,i_s,j_1,...,j_r,k_1,...,k_r}$ are rational numbers.
\end{theorem}

Theorem~\ref{deconv} is proved in Appendix~\ref{deconvproof},
and states exactly when we can hope for performing deconvolution,
either by inferring on the spectrum of ${\bf D}_i(N)$, or on the spectrum or the phase distribution of ${\bf V}_i$ from (\ref{moregeneral2}).
The proof will also state concrete expressions for the mixed moments which parallel the expressions of~\cite{ryandebbah:vandermonde1},
and also summarize the algorithm needed to compute these expressions, as performed by the implementation.
The implementation is thus {\em moment-based}, in that it computes the moments as defined in (\ref{moment}),
from the moments of the input matrices. We do not know any other methods than that of moments to infer on the spectra of such matrices,
since other analytical tools have not been developed yet.

As an example, Theorem~\ref{deconv} states that 
\begin{equation} \label{examplehere}
  \mathrm{tr}\left(\left( ({\bf V}_1+{\bf V}_2+\cdots)^H({\bf V}_1+{\bf V}_2+\cdots)   \right)^p\right),
\end{equation}
which characterize the singular law of a sum of independent Vandermonde matrices, 
depend only on the moments when the ${\bf V}_i$ are independent with the same phase distribution.
When the phase distributions are different, however, the same can not be said.
The final observation in Theorem~\ref{deconv} about the polynomial form of the mixed moment is also important,
since it is a property shared with freeness.
Although (\ref{mostgeneral}) is seen not to be multi-linear in the moments in general,
several of the particular convolutions we consider will be seen to have such a multi-linearity property.

In the following, we state expressions for the convolutions 1)-4) on the form (\ref{mostgeneral}).
Their proofs will be apparent from the proof of Theorem~\ref{deconv}, and can be found in Appendix~\ref{appendixteo012}.
The aspect ratio $c$ will be handled in a particular way in these results, so that it is applied outside the algorithm itself.
The results are stated so that it is possible to turn them around for "deconvolution":
for instance, from the moments of ${\bf D}(N){\bf V}^H{\bf V}$, one can infer on the moments of ${\bf D}(N)$.
The application of the theorems in terms of deconvolution is certainly as important as the limit results themselves,
since it enables us to infer on the parameters in an underlying model (here represented by ${\bf D}(N)$ and ${\bf V}$).
The accompanying implementation of this paper also supports deconvolution.

As for the convolutions 2), this form is not compatible with the form (\ref{moregeneral2}) due to the placement of the ${\bf D}(N)$.
We will therefore not handle this operation, only state in Appendix~\ref{appendixteo012} why one in this case can't expect that the result
only depends on the spectra of ${\bf D}(N)$ and ${\bf V}$.

All formulas in the following are generated by the accompanying software implementation, which is gone through in Section~\ref{theoremimpl}.
Implementation details pertaining to the different convolutions are gone through in Appendix~\ref{appendixteo012}.
Note that the software implementation is capable not only of generating the listed mathematical formulas for the convolutions,
but also to perform the computations numerically, as would be needed in real-time applications.

% Case 1

\subsection{The convolutions $\lim_{N\rightarrow\infty} {\bf D}(N){\bf V}^H{\bf V}$ and $\lim_{N\rightarrow\infty} ({\bf D}(N) + {\bf V}^H{\bf V})$}
In Theorem 1 of~\cite{ryandebbah:vandermonde1}, the moments $\lim_{N\rightarrow\infty} \mathrm{tr}\left(\left({\bf D}(N){\bf V}^H{\bf V}\right)^n\right)$
were expressed in terms of the integrals 
\begin{equation} \label{ikdef}
  I_{k,\omega}=(2\pi)^{k-1}\int_0^{2\pi} p_{\omega}(x)^k,
\end{equation}
$p_{\omega}$ being the density of the phase distribution.
These again determine the moments of ${\bf V}^H{\bf V}$ uniquely ((13) and (20) in~\cite{ryandebbah:vandermonde1}),
so that, indeed, the moments of the matrices (\ref{firsttypes}) depend only on the spectra of the input matrices.
This gives the following result for the multiplicative convolution in 1):

\begin{theorem} \label{teo0}
Assume that ${\bf V}$ has a phase distribution with continuous density,
\begin{eqnarray}
  V_n &=& \lim_{N\rightarrow\infty} \mathrm{tr} \left( \left( {\bf V}^H{\bf V} \right)^n\right) \label{vdef} \\
  D_n &=& c\lim_{N\rightarrow\infty} \mathrm{tr} \left( {\bf D}(N)^n \right) \label{dndef} \\
  M_n &=& c\lim_{N\rightarrow\infty} \mathrm{tr}\left( \left({\bf D}(N){\bf V}^H{\bf V}\right)^n\right), \label{mndef}
\end{eqnarray}
where $c=\lim_{N\rightarrow\infty} \frac{L}{N}$. Then we have that
\begin{eqnarray*}
  M_{1} &=& D_{1}\\ 
  M_{2} &=& D_{2} - D_{1}^{2} + D_{1}^{2}V_{2}\\ 
  M_{3} &=& D_{3} - 3D_{2}D_{1} + 3D_{2}D_{1}V_{2}\\ 
        & & + 2D_{1}^{3} - 3D_{1}^{3}V_{2} + D_{1}^{3}V_{3}\\ 
  M_{4} &=& D_{4} - \frac{8}{3}D_{2}^{2} + \frac{8}{3}D_{2}^{2}V_{2} - 4D_{3}D_{1}\\ 
        & & + 4D_{3}D_{1}V_{2} + 12D_{2}D_{1}^{2} - 18D_{2}D_{1}^{2}V_{2}\\ 
        & & + 6D_{2}D_{1}^{2}V_{3} - \frac{19}{3}D_{1}^{4} + \frac{34}{3}D_{1}^{4}V_{2}\\ 
        & & - 6D_{1}^{4}V_{3} + D_{1}^{4}V_{4}
\end{eqnarray*}

where all coefficients are rational numbers.
Also, whenever $\{M_n\}_{1\leq n\leq k}$ are known, and $\{V_n\}_{1\leq n\leq k}$(or $\{D_n\}_{1\leq n\leq k}$) also are known,
then $\{D_n\}_{1\leq n\leq k}$(or $\{V_n\}_{1\leq n\leq k}$) are uniquely determined.
\end{theorem}

The proof of Theorem~\ref{teo0} can be found in Appendix~\ref{appendixteo012}.
Restricting to uniform phase distribution we get the following result, also generated by the implementation.

\begin{corollary} \label{firstcorollary}
When ${\bf V}$ has uniform phase distribution, we have that
\begin{eqnarray*}
  M_{1} &=& D_{1}\\
  M_{2} &=& D_{2} + D_{1}^{2}\\
  M_{3} &=& D_{3} + 3D_{2}D_{1} + D_{1}^{3}\\
  M_{4} &=& D_{4} + \frac{8}{3}D_{2}^{2} + 4D_{3}D_{1} + 6D_{2}D_{1}^{2} + D_{1}^{4}
\end{eqnarray*}

\end{corollary}

The additive convolution in 1) can be split into sums of many terms similar to (\ref{mndef}), and for each term,
the results of~\cite{ryandebbah:vandermonde1} can be applied.
We obtain the following result, also proved in Appendix~\ref{appendixteo012}:

\begin{theorem} \label{teo1}
Assume that has a phase distribution with continuous density,
\begin{eqnarray*}
  M_n &=& c\lim_{N\rightarrow\infty} \mathrm{tr}\left( \left({\bf D}(N) + {\bf V}^H{\bf V}\right)^n\right),
\end{eqnarray*}
where $c=\lim_{N\rightarrow\infty} \frac{L}{N}$. 
With $V_n$ as in (\ref{vdef}) and $D_n$ as in (\ref{dndef}), we have that
\begin{eqnarray*}
  M_{1} &=& D_{1}+ 1\\
  M_{2} &=& D_{2} + 2D_{1} + V_{2}\\
  M_{3} &=& D_{3} + 3D_{2} + 3D_{1}V_{2} + V_{3}\\
  M_{4} &=& D_{4} + 4D_{3} + 2D_{2} + 4D_{2}V_{2}\\ 
        & & - 2D_{1}^{2} + 2D_{1}^{2}V_{2} + 4D_{1}V_{3} + V_{4}
\end{eqnarray*}

where all coefficients are rational numbers.
Also, whenever $\{M_n\}_{1\leq n\leq k}$ are known, and $\{V_n\}_{1\leq n\leq k}$(or $\{D_n\}_{1\leq n\leq k}$) also are known,
then $\{D_n\}_{1\leq n\leq k}$ (or $\{V_n\}_{1\leq n\leq k}$) are uniquely determined.
\end{theorem}

Restricting to uniform phase distribution we get another specialized result:

\begin{corollary} \label{secondcorollary}
When ${\bf V}$ has uniform phase distribution, we have that
\begin{eqnarray*}
  M_{1} &=& D_{1} + 1\\
  M_{2} &=& D_{2} +2D_{1} + 2\\
  M_{3} &=& D_{3} +3D_{2} +6D_{1} + 5\\
  M_{4} &=& D_{4} +4D_{3} +10D_{2} +2D_{1}^{2} +20D_{1} + \frac{44}{3}
\end{eqnarray*}

\end{corollary}

% Case 3

\subsection{The convolutions $\lim_{N\rightarrow\infty} {\bf V}_1^H{\bf V}_1 {\bf V}_2^H{\bf V}_2$ and $\lim_{N\rightarrow\infty} ({\bf V}_1^H{\bf V}_1 + {\bf V}_2^H{\bf V}_2)$}
The following result says that the convolution 3) only depends on the spectra of the input matrices:
\begin{theorem} \label{teo4}
Assume that ${\bf V}_1$ and ${\bf V}_2$ are independent Vandermonde matrices where the phase distributions have continuous densities, and set
\begin{eqnarray}
  V_1^{(n)} &=& \lim_{N\rightarrow\infty} \mathrm{tr} \left( \left( {\bf V}_1^H {\bf V}_1 \right)^n\right) \nonumber \\
  V_2^{(n)} &=& \lim_{N\rightarrow\infty} \mathrm{tr} \left( \left( {\bf V}_2^H {\bf V}_2 \right)^n\right) \nonumber \\
  M_n       &=& \lim_{N\rightarrow\infty} \mathrm{tr}\left(({\bf V}_1^H{\bf V}_1 {\bf V}_2^H{\bf V}_2)^n\right) \label{multeq1} \\
  N_n       &=& \lim_{N\rightarrow\infty} \mathrm{tr}\left(({\bf V}_1^H{\bf V}_1 + {\bf V}_2^H{\bf V}_2)^n\right) \label{multeq2}
\end{eqnarray}
$M_n,N_n$ are completely determined by $V_2^{(i)},V_3^{(i)},...$,
and the aspect ratios $c_1=\lim_{N_1\rightarrow\infty} \frac{L}{N_1}, c_2=\lim_{N_2\rightarrow\infty} \frac{L}{N_2}$.
Moreover, $M_n,N_n$ are higher degree polynomials in the $V_2^{(i)},V_3^{(i)},...$ on the form (\ref{mostgeneral}).
Also, whenever $\{M_n\}_{1\leq n\leq k}$ (or $\{N_n\}_{1\leq n\leq k}$) are known, and $\{V_1^{(n)}\}_{1\leq n\leq k}$ also are known,
then $\{V_2^{(n)}\}_{1\leq n\leq k}$ are uniquely determined.
\end{theorem}

The proof can be found in Appendix~\ref{appendixteo012}.
Due to the complexity in the expressions , we do not state formulas
for the first moments in Theorem~\ref{teo4}.

Interestingly, since the joint distribution of $\{ {\bf V}^H{\bf V},{\bf D}(N)\}$ is not multi-linear in the moments of ${\bf D}(N)$,
while the joint distribution of  $\{ {\bf V}_1^H{\bf V}_1,{\bf V}_2^H{\bf V}_2\}$ is,
it is seen that the joint distributions are different in the two cases, even if the moments of the component matrices are the same.

% Case 4 when equal, product case

\subsection{The convolution $\lim_{N\rightarrow\infty} {\bf V}_1{\bf V}_1^H {\bf V}_2{\bf V}_2^H$ when the matrices have equal phase distribution}
When the phase distributions are different, Theorem~\ref{deconv} explains that the moments of ${\bf V}_1{\bf V}_1^H {\bf V}_2{\bf V}_2^H$
are not necessarily expressible in terms of the moments of the component matrices.
This is, however, the case when the phase distributions are equal.
We thus have the following result, which proof can be found in Appendix~\ref{appendixteo012}:

\begin{theorem} \label{teo2}
Assume that ${\bf V}_1$ and ${\bf V}_2$ are independent Vandermonde matrices with the same phase distribution, and that this has a continuous density, and set
\begin{eqnarray*}
  V_n &=& \lim_{N\rightarrow\infty} \mathrm{tr} \left( \left( {\bf V}_i^H {\bf V}_i \right)^n\right) \\
  M_n &=& \lim_{N\rightarrow\infty} \mathrm{tr}\left(({\bf V}_1^H{\bf V}_2 {\bf V}_2^H{\bf V}_1)^i\right).
\end{eqnarray*}
Then we have that
\begin{eqnarray*}
  M_{1} &=& - 1 + V_{2}\\
  M_{2} &=& - 3 + 6V_{2} - 4V_{3} + V_{4}\\
  M_{3} &=& - 58 + 123V_{2} - 96V_{3} + 39V_{4} - 9V_{5} + V_{6}\\
  M_{4} &=& - \frac{21532}{5} + \frac{410726}{45}V_{2} - \frac{321191}{45}V_{3} + \frac{44516}{15}V_{4}\\ 
        & & - 772V_{5} + 136V_{6} - 16V_{7} + V_{8}
\end{eqnarray*}

\end{theorem}

Restricting to uniform phase distribution we get another specialized result:

\begin{corollary}
When ${\bf V}_1$ and ${\bf V}_2$ have uniform phase distribution, we have that
\begin{eqnarray*}
  M_{1} &=& 1\\
  M_{2} &=& 2\\
  M_{3} &=& 5\\
  M_{4} &=& \frac{44}{3}
\end{eqnarray*}

\end{corollary}

% Case 4, when different, summation case

\subsection{The convolution $\lim_{N\rightarrow\infty} \left({\bf V}_{\omega_1}^{(1)} \left({\bf V}_{\omega_1}^{(1)}\right)^H + {\bf V}_{\omega_2}^{(2)} \left( {\bf V}_{\omega_2}^{(2)}\right)^H\right)$} \label{counterex}
${\bf V}^H{\bf V}$ can be viewed as the sample covariance matrix of the random vector $(1,e^{-j\omega},...,e^{-j(N-1)\omega})$. 
A similar interpretation of the convolution
$\left({\bf V}_{\omega_1}^{(1)} \left({\bf V}_{\omega_1}^{(1)}\right)^H + {\bf V}_{\omega_2}^{(2)} \left( {\bf V}_{\omega_2}^{(2)}\right)^H\right)$
is thus as a sample covariance matrix of a random vector of the same type, 
but where the phase distribution is $\omega_1$ parts of the time, and $\omega_2$ the rest of the time.
This convolution does not satisfy the requirement $\sigma\geq [0,1]_n$ from Theorem~\ref{deconv},
so there is no guarantee that the result only depends on the spectra of the input matrices.
It will be apparent from Theorem~\ref{teonew} below that the dependence is, indeed, on more than just these spectra:
Knowledge about the phase distributions is also required, and we will in fact interpret this convolution instead as an operation on phase distributions.

Consider first two independent Vandermonde matrices ${\bf V}_{\omega,c_1}^{(1)}$, ${\bf V}_{\omega,c_2}^{(2)}$
with an equal number of rows $N$ and with a common phase distribution $\omega$.
By stacking ${\bf V}_{\omega,c_1}^{(1)}$, ${\bf V}_{\omega,c_2}^{(2)}$ horizontally into one larger matrix,
it is straightforward to show that the distribution of
\begin{equation} \label{trivialresult}
  {\bf V}_{\omega,c_1}^{(1)} \left({\bf V}_{\omega,c_1}^{(1)}\right)^H + {\bf V}_{\omega,c_2}^{(2)} \left( {\bf V}_{\omega,c_2}^{(2)}\right)^H
\end{equation}
equals that of ${\bf V}_{\omega,c_1+c_2}{\bf V}_{\omega,c_1+c_2}^H$.
This case when the phase distributions are equal is therefore trivial.

When ${\bf V}_{\omega_1,c}^{(1)}$, ${\bf V}_{\omega_2,c}^{(2)}$ are independent with the same number of rows, but with different phase distributions, 
computing the distribution of
\begin{equation} \label{firstcasehere}
  {\bf V}_{\omega_1,c_1}^{(1)} \left({\bf V}_{\omega_1,c_1}^{(1)}\right)^H + {\bf V}_{\omega_2,c_2}^{(2)} \left( {\bf V}_{\omega_2,c_2}^{(2)}\right)^H
\end{equation}
seems, however, to be more complex. The following result explains that, at least in the limit, the situation is simpler. 
There the sum can be replaced by another Vandermonde matrix, whose phase distribution can be constructed in a particular way from the original ones:

\begin{theorem} \label{teonew}
  Let ${\bf V}_{\omega_1,c_1}$ and ${\bf V}_{\omega_2,c_2}$ be independent $N\times L_1$, $N\times L_2$ random Vandermonde matrices
  with phase distributions $\omega_1$, $\omega_2$, respectively, and with aspect ratios
  $c_1=\lim_{N\rightarrow\infty} \frac{L_1}{N}$, $c_2=\lim_{N\rightarrow\infty} \frac{L_2}{N}$, respectively.
  Then the limit distribution of
  \begin{equation} \label{secondtype}
    {\bf V}_{\omega_1,c_1} {\bf V}_{\omega_1,c_1}^H + {\bf V}_{\omega_2,c_2} {\bf V}_{\omega_2,c_2}^H
  \end{equation}
  equals that of
  \begin{equation} \label{alternative}
    {\bf V}_{\omega_1\ast_{c_1,c_2}\omega_2,c_1+c_2} {\bf V}_{\omega_1\ast\omega_2,c_1+c_2}^H,
  \end{equation}
  where $\omega_1\ast_{c_1,c_2}\omega_2$ denotes the phase distribution with density $\frac{1}{c_1+c_2}(c_1p_{\omega_1}+c_2p_{\omega_2})$,
  where $p_{\omega_1},p_{\omega_2}$ are the densities of the phase distributions $\omega_1,\omega_2$.
\end{theorem}

The proof of Theorem~\ref{teonew} can be found in Appendix~\ref{appendixteonew}.
The result is only asymptotic, meaning that the mean eigenvalue distribution for finite $N$ of the two mentioned matrices are in fact different.
This can be seen by setting $L=N=2$, and observing that the distribution of
$\frac{1}{2}\left(e^{j\omega_1}+e^{j\omega_2}\right)$ is in general different from that of $e^{\omega_1\ast_{1,1}\omega_2}$.
No trivial proof for Theorem~\ref{teonew} is thus known, since the strategy of stacking the Vandermonde matrices
(from the reasoning for (\ref{trivialresult})) will not work.

Theorem~\ref{teonew} says that one depends on knowledge about the phase distributions for Convolution 4).
To verify this, set $\omega_1$ and $\omega_2$ equal to the uniform distributions on $[0,\pi)$,
and then change $\omega_2$ to the uniform distribution on $[\pi,2\pi)$.
The phase distributions here give the same moments (since they are shifted versions).
However, the two versions of $\frac{1}{2}(p_{\omega_1}+p_{\omega_2})$ give phase distributions with different moments, since we get the
uniform distribution on $[0,\pi)$ in the first case, and the uniform distribution on $[0,2\pi)$ in the second case: 
the moments of these are different, since the uniform distribution on $[0,2\pi)$
minimizes the moments of Vandermonde matrices~\cite{ryandebbah:vandermonde1}.
For the same reason, Theorem~\ref{teonew} says that the moments of (\ref{secondtype}) are minimized when
$\omega_1\ast_{c_1,c_2}\omega_2$ equals the uniform distribution.

\subsection{Hankel and Toeplitz matrices}
\cite{ryandebbah:vandermonde1} states that the moments of ${\bf V}^H{\bf V}$ can be expressed in terms of volumes of certain convex polytopes.
It turns out that the moments of Hankel, Markov and Toeplitz matrices can be expressed in terms of a subset of these polytopes~\cite{paper:brycdembojiang},
so that we can use the same strategy to compute the moments of these matrices also.
The proof of the following theorem relating to the moments of Toeplitz matrices is therefore explained in Appendix~\ref{appendixteo012}.

\begin{theorem} \label{teo3}
Define the Toeplitz matrix
\[
  {\bf T}_n = \frac{1}{\sqrt{n}} \left( \begin{array}{cccccc} X_0 & X_1 & X_2 & \cdots & X_{n-2} & X_{n-1} \\
                                           X_1 & X_0 & X_1 &        &         & X_{n-2} \\
                                           X_2 & X_1 & X_0 &        & \ddots  & \vdots \\
                                           \vdots &  &     & \ddots &         & X_2 \\
                                           X_{n-2} & &     &        & X_0     & X_1 \\
                                           X_{n-1} & X_{n-2} & \hdots & X_2 & X_1 & X_0 \end{array} \right),
\]
where $X_i$ are i.i.d., real-valued random variables with variance $1$.
Let $M_i$ be the $2i$'th asymptotic moment of ${\bf T}_n$ (the odd moments vanish). These moments are given by
\begin{eqnarray*}
  M_{1} &=& 1\\
  M_{2} &=& \frac{8}{3}\\
  M_{3} &=& 11\\
  M_{4} &=& \frac{1435}{24}
\end{eqnarray*}

\end{theorem}

A similar result for Hankel matrices also holds:
\begin{theorem} \label{teohankel}
Define the Hankel matrix
\[
  {\bf H}_n = \frac{1}{\sqrt{n}} \left( \begin{array}{cccccc}
                                           X_1 & X_2 & \cdots & \cdots & X_{n-1} & X_n \\
                                           X_2 & X_3 &        &        & X_n     & X_{n+1} \\
                                           \vdots &  &        &        & X_{n+1} & X_{n+2} \\
                                                  &  &        & \ddots &         & \\
                                           X_{n-2} & X_{n-1} & &       &         & \vdots \\
                                           X_{n-1} & X_n & & & X_{2n-3} & X_{2n-2} \\
                                           X_n & X_{n+1} & \cdots & \cdots & X_{2n-2} & X_{2n-1}.
                                        \end{array} \right),
\]
where $X_i$ are i.i.d., real-valued random variables with variance $1$.
Let $M_i$ be the $2i$'th asymptotic moment of ${\bf H}_n$ (the odd moments vanish). These moments are given by
\begin{eqnarray*}
  M_{1} &=& 1\\
  M_{2} &=& \frac{8}{3}\\
  M_{3} &=& 14\\
  M_{4} &=& 100
\end{eqnarray*}

\end{theorem}

Similar results can also be written down for Markov matrices, but these expressions are skipped.
It seems that expressions for the joint distribution of Hankel and Toeplitz matrices and matrices ${\bf D}(N)$ on the same form as before do not exist, 
meaning that the mixed moments may not exist, or that they depend on more than the spectra of the component matrices. 
The details of this are also skipped.

\subsection{Generalizations to almost sure convergence}
Up to now, we have only shown convergence in distribution for the different convolutions and mixed moments.
The same results also hold when we replace convergence in distribution with almost sure convergence in distribution.
We summarize this in the following result:

\begin{theorem} \label{almostsuretheorem}
Assume that the matrices ${\bf D}_i(N)$ have a joint limit distribution as $N\rightarrow\infty$,
and that ${\bf V}_1,{\bf V}_2,...$ are independent, with continuous phase distributions.
Any combination of matrices on the form (\ref{moregeneral}) converges almost surely in distribution, whenever the matrix product is well-defined and square.
\end{theorem}

The proof of Theorem~\ref{almostsuretheorem} can be found in Appendix~\ref{appendixteoas}.
In particular, the matrices we have considered in our convolution operations, such as
${\bf V}_1^H{\bf V}_1{\bf V}_2^H{\bf V}_2$, ${\bf V}_1^H{\bf V}_1 + {\bf V}_2^H{\bf V}_2$,
all converge almost surely in distribution.

\subsection{Generalized Vandermonde matrices}
We have not considered generalized Vandermonde matrices up to now,
i.e. matrices were the columns in ${\bf V}$ are not uniform distributions of powers~\cite{norberg,ryandebbah:vandermonde1}.
Although similar results can also be stated for these matrices, we only explain how they will differ.

In case of uniform power distribution, the column sum of (\ref{vandermonde}) is
\begin{equation} \label{uniformpower}
  \frac{1-e^{j Nx}}{1-e^{j x}},
\end{equation}
and this is substituted into the integrand of the expression defining the Vandermonde mixed moment expansion coefficients (see Appendix~\ref{deconvproof}).
For generalized Vandermonde matrices, one can also define these coefficients~\cite{ryandebbah:vandermonde1},
the difference being that one replaces the sum of the powers (\ref{uniformpower}) with a different function, 
and requires that the function has the property proved in Lemma~\ref{appendixlemma} in Appendix~\ref{deconvproof}. 
The details for computing the mixed moments (\ref{moregeneral2}) go otherwise the same way as the expressions in Appendix~\ref{deconvproof},
with the exception that we have different values for the Vandermonde mixed moment expansion coefficients.
However, the integrals defining these coefficients may be hard to compute for a non-uniform power distribution, even for the case of uniform
phase distribution, since Fourier-Motzkin elimination (see Section~\ref{theoremimpl}) can be applied only in the case of uniform power- and phase distribution.

We conjecture that Theorem~\ref{teonew} holds also for general power distributions. It is likely that a similar calculation as in Appendix~\ref{appendixteonew}
can prove this, but we do not go into details on this.

\section{Software implementation} \label{theoremimpl}
In this section, we will repeatedly refer to the implementation~\cite{eurecom:vandermondeimpl2}, which contains all code needed to verify all results in this paper. 
Implementations therein have two purposes: 
\begin{enumerate}
  \item to generate the exact coefficients in the formulas in this paper (generated directly in latex),
  \item to compute the convolution with a given set of moments numerically. 
\end{enumerate}
In Appendix~\ref{deconvproof}, we explain why iteration through partitions and Fourier-Motzkin elimination are two main things needed in the implementation. 
In this section, we will explain how these tasks can be implemented efficiently. 

\subsection{Reducing the complexity in iterating over partitions}
Formulas in~\cite{ryandebbah:vandermonde1} and in this paper sum over sets of partitions. 
Iterating over partitions is very time-consuming, and must therefore be performed efficiently. 
There are several ways how this 
can be performed\footnote{The implementation in this paper uses an implementation~\cite{eurecom:vandermondeimpl2} 
which lists all partitions of $n$ elements with a given number of blocks}.
It turns out that one can reduce the number of partitions needed for computations considerably. 

Assume that ${\bf V}$ has uniform phase distribution, and consider
\begin{equation} \label{standardform}
  \mathrm{tr}\left( {\bf V}^H{\bf V}\cdots {\bf V}^H{\bf V} \right).
\end{equation}
To compute (\ref{standardform}), we traverse all partitions. For each partition an equation system is constructed, 
and the partition contributes with the volume of the corresponding solution set to the equation system in (\ref{standardform}).
The following observations~\cite{paper:nordio1,ryandebbah:vandermonde1} simplifies this computation:
\begin{itemize}
  \item If a block is a singleton, then the corresponding volume is the same as that of the partition with that block removed. 
    By using this observation repeatedly, we obtain that any noncrossing partition gives 1 in volume contribution.
  \item If a block contains two successive elements, then the corresponding volume is the same as that of the partition with any one of the two elements removed
  \item If a partition is a cyclic shift of another, then the corresponding volumes are the same.
\end{itemize}
These observations can reduce the number of the computations dramatically. 
To make precise how these observations can be used, we state two definitions:
\begin{definition}
A partition $\pi$ is said to be alternating if $i$ and $i+1$ (where the sum is taken cyclically mod $n$) are in different blocks for all $i$, 
and no blocks in $\pi$ are singletons. 
The alternating partition obtained by removing all singleton blocks and all successive elements in all blocks incrementally 
is called the standard form of the partition.
The set of alternating partitions of $\{ 1,...,n\}$ with $k$ blocks is denoted ${\cal A}(n,k)$.
\end{definition}
\begin{definition} \label{eqvdef}
We say that two partitions are equivalent whenever one is a cyclic shift (with a fixed number of elements) of the other.
\end{definition}

Note that the standard form of any noncrossing partitions is the empty partition. 
The first two observations above say that computations only need to be performed for alternating partitions
(since any partition can be reduced to an alternating one in standard form), 
while the third observation says that computations are only needed for one representative in each equivalence class, 
with equivalence defined as in Definition~\ref{eqvdef}. 
For instance, there are $678570$ partitions of $\{ 1,...,11\}$. 
The number of alternating partitions of the same set is $4427$. 
The number of equivalence classes of alternating partitions is $715$. 

The moments of Vandermonde matrices can thus be computed by iterating over the smaller set of cyclic equivalence classes of alternating partitions. 
This iteration can be accomplished with a 
computer program\footnote{It is not obvious how the observations can be applied in an efficient implementation. The implementation~\cite{eurecom:vandermondeimpl2}
first generates all partitions, and then picks out those which have the alternating property and no singleton blocks}. 
We also need to keep track of the size of each equivalence class of alternating partitions. 
This is done by a program which efficiently hashes all partitions. 
This is a computationally intensive process, but which needs to be done only once for the required number of moments.

\subsection{Constructing linear equation systems}
For a partition $\rho=\{ W_1,...,W_r\}\in{\cal P}(n)$,~\cite{ryandebbah:vandermonde1} relates (\ref{standardform}) 
to the corresponding volume of the solution set of the equations 
\begin{eqnarray}
  \sum_{k\in W_1} x_{k-1} &=& \sum_{k\in W_1} x_k \nonumber \\
  \sum_{k\in W_2} x_{k-1} &=& \sum_{k\in W_2} x_k \nonumber \\
  \vdots                  & & \vdots \nonumber \\
  \sum_{k\in W_r} x_{k-1} &=& \sum_{k\in W_r} x_k, \label{volumeeq}
\end{eqnarray}
where all variables are constrained to lie between $0$ and $1$. 
$\rho$ reflects how the $\omega_i$ are grouped into independent sets of variables:
The left sides in (\ref{volumeeq}) represent the ${\bf V}^H$-terms in the entries of the matrix product ${\bf V}^H{\bf V}$, 
whereas the right sides represent the ${\bf V}$-terms in the same matrix product. 
Equations of the form (\ref{volumeeq}) also apply to the more general form~\cite{ryandebbah:vandermonde1}
\begin{equation} \label{standardform2}
  \mathrm{tr}\left( {\bf V}_{i_1}{\bf V}_{i_1}^H\cdots {\bf V}_{i_k}{\bf V}_{i_k}^H \right),
\end{equation}
where ${\bf V}_1,{\bf V}_2,...$ are independent and with uniform phase distribution.

In Appendix~\ref{deconvproof}, it is shown that in order to express the arbitrary mixed moments of Vandermonde matrices (independent or not), 
we need to solve systems similar to (\ref{volumeeq}), with the difference that different number of variables may appear on the left and right hand sides. 
Note that the volume of the solution set of (\ref{volumeeq}) is always a rational number. 
This enables our implementation to generate exact formulas. 
For Toeplitz and Hankel matrices, it turns out that a subset of these equation systems serve the same role in order to compute their moments.

\subsection{Solving the linear equation systems}
In all cases of Toeplitz, Hankel, and Vandermonde, the coefficient matrix of the equations we construct has rank $r-1$ ($r$ being the number of blocks), 
and we need to find the number of solutions. Since we also have the constraints that $0\leq x_i\leq 1$, 
this really corresponds to finding all solutions to a set of linear inequalities. 
A much preferred method for doing so is Fourier-Motzkin elimination~\cite{paper:dahl1}. 
The first step before we perform this elimination would be to bring the equations into a standard form.
We do this by expressing the $r-1$ pivot variables (after row reduction) by means of the free variables. 
Since all variables are between $0$ and $1$ (which are split into two inequalities), our equations are
\begin{equation} \label{likninger0}
\begin{array}{rrr}
  \sum_{j=1}^{n-r+1} a_{1j}x_j      &\leq& 1 \\
  \sum_{j=1}^{n-r+1} -a_{1j}x_j     &\leq& 0 \\
  \sum_{j=1}^{n-r+1} a_{2j}x_j      &\leq& 1 \\
  \sum_{j=1}^{n-r+1} -a_{2j}x_j     &\leq& 0 \\
  \vdots                            &    & \vdots \\
  \sum_{j=1}^{n-r+1} a_{(r-1)j}x_j  &\leq& 1 \\
  \sum_{j=1}^{n-r+1} -a_{(r-1)j}x_j &\leq& 0 \\
   x_1                              &\leq& 1 \\
  -x_1                              &\leq& 0 \\
   x_2                              &\leq& 1 \\
  -x_2                              &\leq& 0 \\
  \vdots                            &    & \vdots \\
  x_{n-r+1}                         &\leq& 1 \\
  -x_{n-r+1}                        &\leq& 0,
\end{array}
\end{equation}
where we have re-indexed the variables so that $x_1,...,x_{n-r+1}$ are the free variables, $x_{n-r+2},...,x_n$ are the pivot variables. 
The coefficients $a_{ij}$ are taken from $-1,0,1$, and are the coefficients we obtain when the pivot variables are expressed in terms of the free variables. 
By reordering the equations, we get what we call the standard form (where the equations are sorted by the first coefficient):
\begin{equation} \label{likninger}
\begin{array}{rrrr}
   x_1    & + \sum_{j=1}^{n-r} b_{1j}x_{j+1}   &\leq& e_1 \\
   \vdots & \vdots                             &    & \vdots \\
   x_1    & + \sum_{j=1}^{n-r} b_{r_1j}x_{j+1} &\leq& e_{r_1} \\
          & \sum_{j=1}^{n-r} c_{1j}x_{j+1}     &\leq& f_1 \\
   \vdots & \vdots                             &    & \\
          & \sum_{j=1}^{n-r} c_{r_2j}x_{j+1}   &\leq& f_{r_2} \\
   -x_1   & + \sum_{j=1}^{n-r} d_{1j}x_{j+1}   &\leq& g_1\\
   \vdots & \vdots                             &    & \vdots \\
   -x_1   & + \sum_{j=1}^{n-r} d_{r_3j}x_{j+1} &\leq& g_{r_3}.
\end{array}
\end{equation}
Fourier-Motzkin elimination now consists of eliminating the first variable, and working on the remaining equations to eliminate variables iteratively. 
Most of the coefficient matrices here are {\em combinatorial matrices} on the same form as those in~\cite{paper:dahl1}. 

Fourier-Motzkin elimination is computationally intensive, in the sense that the number of inequalities grow rapidly during elimination. 
Our aim is to compute the volume of the solution set rather than finding specific solutions. 
The volume can be split into many smaller disjoint parts, 
each part corresponds to a choice of minimum ({\tt min}) for the first equations, and a choice of maximum ({\tt max}) for the last equations
Each part corresponds to the solution of a set of equations with one less variable. 
More precisely, let the equations in (\ref{likninger}) have coefficient vectors $B_1,...,B_{r_1},C_1,...,C_{r_2},D_1,...,D_{r_3}$, so that 
\[
\begin{array}{rrrl}
  B_i &=& (1 ,b_{i1},...,b_{i(n-r)},e_i) & 1\leq i\leq r_1 \\
  C_i &=& (0 ,c_{i1},...,c_{i(n-r)},f_i) & 1\leq i\leq r_2 \\
  D_i &=& (-1,d_{i1},...,d_{i(n-r)},g_i) & 1\leq i\leq r_3.
\end{array}
\]
Each choice of {\tt min}, {\tt max}, with $1\leq \mbox{\tt min}\leq r_1$, $1\leq \mbox{\tt max}\leq r_3$ 
gives rise to a volume described by the solution to the set of equations
\[
\begin{array}{lll}
B_k-B_{\mbox{\tt min}}                & 1\leq k\leq r_1 & k\neq \mbox{\tt min} \\
D_k-D_{\mbox{\tt max}}                & 1\leq k\leq r_3 & k\neq \mbox{\tt max} \\
D_{\mbox{\tt max}}+B_{\mbox{\tt min}} &                 & \\
C_k                                   & 1\leq k\leq r_2,& \\
\end{array}
\]
where the equations are described by row vectors as above. 
There are $r_1-1+r_3-1+1+r_2=r_1+r_2+r_3-1$ equations here, which is one less equation than what we started with. 
Note that the first element is zero in all these equations, so that the first column can be removed in the coefficient matrix. 
Therefore, the original system has been reduced to one with one equation less and one less variable. 
There may be more zero leading columns also, and all these can be removed. When the leading column is nonzero, the rows are sorted so that 
we get a new system on the form (\ref{likninger}), and the procedure continues. 
In the process, the choice of {\tt max} and {\tt min} have decided the lower and upper integral bounds for the $x_1$-variable. 
These are stored, and after all Fourier-Motzkin elimination we have a full set of integral bounds, 
and the corresponding volume is computed by integrating over these bounds. 
This can be implemented easily~\cite{eurecom:vandermondeimpl2}, since integration over a volume with integral bounds 
which are linear in the variables can be defined in terms of simple row operations and integration by parts. 

\subsection{Optimizations for Fourier-Motzkin elimination}
The challenge in computing the volumes of the solution sets in Fourier-Motzkin elimination lies in
that there are many eliminations which need to be peformed, and we do this for every partition in a large set of partitions. 
The Fourier-Motzkin elimination steps themselves can be stored and reused, 
but this of little help since we have to keep track of the corresponding integral bounds for the solution sets. 
There are however, a couple of optimizations which can be used during elimination:
\begin{itemize}
  \item if both row and the negative of that row are present as an equation, the solution set is empty, so that we can stop elimination
  \item Duplicate rows can be deleted
  \item Rows where only the last elements differ can be merged. 
\end{itemize}

\section{Simulations} \label{simulations}
Results in this paper have been concerned with finding the spectral limit distribution from those of the input matrices.
However, in practice, one has a certain model where one or more parameters are unknown, one observes output from that model,
and would like to infer on the parameters of the model. The strengths in the results of this paper lie in that this
kind of "deconvolution" is made possible to infer on the parameters of various models. As an example,
\begin{enumerate}
  \item From observations of the form ${\bf D}(N){\bf V}^H{\bf V}$ or ${\bf D}(N)+{\bf V}^H{\bf V}$,
    one can infer on either the spectrum of ${\bf D}(N)$, or the spectrum or phase distribution of ${\bf V}$, when exactly one of these is unknown.
  \item From observations of the form ${\bf V}_1^H{\bf V}_1{\bf V}_2^H{\bf V}_2$ or ${\bf V}_1^H{\bf V}_1+{\bf V}_2^H{\bf V}_2$, one can infer on
    the spectrum or phase distribution of one of the Vandermonde matrices, when one of the Vandermonde matrices is known.
\end{enumerate}
Moreover, the complexity in this inference is dictated by the number of moments considered.
We do not go into depths on all the different types of deconvolutions made possible, 
only sketch a very simple example of inference as in 1). 
The other types of deconvolution go similarly, since the implementation supports each of them through functions with similar signatures.
The example only makes an estimate of the first lower order moments of the component matrix ${\bf D}(N)$. 
These moments can give valuable information:
in cases where it is known that there are few distinct eigenvalues, and the multiplicities are known, 
only some lower order moments are needed in order to get an estimate of these eigenvalues. 
We remark that this kind of deconvolution can be improved by further development of a second order theory for Vandermonde matrices.

In Figure~\ref{fig:vandaddsim}, we have, for Vandermonde matrices of size $N\times L$ with $L=N$, and for increasing $N$,
formed $10$ observations of the form ${\bf D}(N){\bf V}^H{\bf V}$. The average of the moments of these observations are then taken,
and a method in the framework~\cite{eurecom:vandermondeimpl2} is applied to get an estimate of the moments of ${\bf D}(N)$.
In the simulation, we have compared the estimate for the second and third moment of ${\bf D}(N)$ obtained by the implementation, 
with the actual second and third moments.
The diagonal matrix ${\bf D}(N)$ is chosen so that the distribution of its eigenvalues is $\frac{1}{3}\delta_{0.5}+\frac{1}{3}\delta_{1}+\frac{1}{3}\delta_{1.5}$,
i.e. $0.5,1,1.5$ are the only eigenvalues, and they have equal probability.
The simulation seems to indicate that the implementation performs better estimation when the matrices grow large,
in accordance with the fact that only an asymptotic result is applied.
\begin{figure}
  \begin{center}
    \epsfig{figure=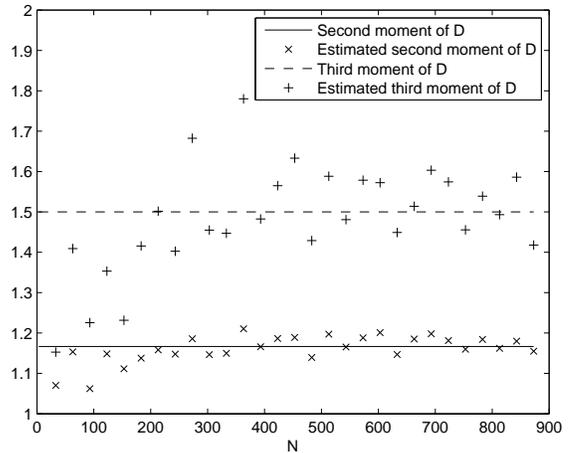,width=0.99\columnwidth}
  \end{center}
  \caption{Estimation of the second and third moment of ${\bf D}(N)$ from the average of $10$ observations of the form ${\bf D}(N){\bf V}^H{\bf V}$,
    for increasing values of $N$. ${\bf V}$ has dimensions $N\times N$.}\label{fig:vandaddsim}
\end{figure}

Although it is difficult to make a full picture of the spectral distribution of ${\bf V}_1^H{\bf V}_1$ (or the phase distribution of ${\bf V}_1$)
from deconvolution on models such as ${\bf V}_1^H{\bf V}_1{\bf V}_2^H{\bf V}_2$
(although the moments in many cases determine the distribution of the eigenvalues~\cite{book:baisilverstein}), such deconvolution can still be useful. 
For instance, from the lower order moments one can to a certain amount say "how far away ${\bf V}_1$ is from having uniform phase distribution", 
since the uniform phase distribution achieves the lowest moments of all Vandermonde matrices~\cite{ryandebbah:vandermonde1}.

\section{Conclusion and further directions}
This contribution has explained how all types of moments in Vandermonde-type expressions can be obtained,
and when one can expect that the moments/spectrum of the result only depend on the moments/spectrum of the input matrices
(which is a requirement for performing deconvolution).
The results can be used to compute the moments of any singular law involving a combination of many independent matrices.
An implementation which is capable of performing these moment computations is also presented,
and moment formulas generated by the implementation were presented. The applications to wireless communications are still under study \cite{ryandebbah:applications}.
We have also described convolution operations on Vandermonde matrices which can not be performed in terms of the spectrum,
but rather in terms of the phase distributions.
We have also expanded known results on convergence of Vandermonde matrices to almost sure convergence.

Interestingly, Vandermonde matrices fit into a framework similar to that of freeness.
Future papers will address a unified framework, where a more general theory which addresses when deconvolution is possible is presented.

It is still an open problem to find exact formulas for any moment of a Vandermonde matrix.
The same applies to identifying these moments as the moments of a certain density.
Future papers may also address how the implementation presented here can be made more efficient.

\appendices

\section{The proof of Theorem~\ref{deconv}} \label{deconvproof}
Let us first assume that all phase distributions are uniform.  
Writing out the matrix product in (\ref{moregeneral2}) we get
\begin{eqnarray}
  & & \sum_{(i_1,...,i_n)} \sum_{(j_1,...,j_n)} \nonumber\\
  & & \hspace{1cm} \prod_{i=1}^{|\sigma_1|} N_i^{-|\sigma_{1i}|/2} L^{-1} \nonumber\\
  & & \hspace{1cm} \times \E ( e^{i_2(\omega_{\sigma_1(1),j_1}-\omega_{\sigma_1(2),j_2})} \times\cdots \nonumber\\
  & & \hspace{2cm}             \times e^{i_1(\omega_{\sigma_1(2n-1),j_n}-\omega_{\sigma_1(2n),j_1})} ) \nonumber\\
  & & \hspace{1cm} \times {\bf D}_1(N)(j_1,j_1) \times\cdots\times {\bf D}_n(N)(j_n,j_n),
\end{eqnarray}
where 
\begin{enumerate}
  \item $1\leq j_1,...,j_n\leq L$ (as in~\cite{ryandebbah:vandermonde1})
  \item $0\leq i_1,...,i_n\leq N_l-1$ for appropriate $l$ (as in~\cite{ryandebbah:vandermonde1}), 
  \item $\sigma_1=\{\sigma_{11},...,\sigma_{1|\sigma_1|}\}$ with $\sigma_{1k}= \{ j | i_j=k \}$,
  \item $\omega_{\sigma_1(i),j_i}$ is the phase for column $j_i$ in the $i$'th matrix entry. 
\end{enumerate}
Define the partition $\pi=\pi(j_1,...,j_n)\in{\cal P}(n)$ by equality of the $j_i$, i.e. $k\sim_{\pi}l$ if and only if $j_k=j_l$. 
Noting that $\omega_{\sigma_1(k),j_k}$, $\omega_{\sigma_1(l),j_l}$ are equal if and only if $\sigma_1(k)=\sigma_1(l)$ and $j_k=j_l$
(if not they are independent), we define
$\rho(\pi)\leq\sigma_1\in{\cal P}(2n)$ as the partition in ${\cal P}(n)$ generated by the relations:
\[
   k \sim_{\rho(\pi)} l \mbox{ if } \left\{ \begin{array}{l} \lfloor k/2\rfloor+1\sim_{\pi}\lfloor l/2\rfloor + 1 \mbox{ and } \\ k\sim_{\sigma_1} l \end{array} \right.
\]
Here $\lfloor x\rfloor$ means the largest whole number less than $x$.  
In other words, $k$ and $l$ are in the same block of $\rho(\pi)$ 
if and only if the corresponding phases $\omega_{\sigma_1(k),j_k}$ and $\omega_{\sigma_1(l),j_l}$ from the $k$'th and $l$'th matrix entries are dependent. 
We will have use for the following relation between $\rho(\pi)$ and $\pi$, 
which will help us to limit our calculations to a certain class of partitions. 
\begin{lemma} \label{pilemma}
The following holds:
\begin{equation} \label{equalityoccurs}
  |\pi| \leq |\rho(\pi)| - r(\pi) + 1, 
\end{equation}
Moreover, both equality and strict inequality can occur in (\ref{equalityoccurs}).
\end{lemma}

\begin{proof}
Since each block in $\pi$ is associated with at least one block in $\rho(\pi)$ by definition, we have that $|\pi|\leq|\rho(\pi)|$.
Moreover, if $\rho_1$ is adjacent to $\rho_2$, they have a $j$-value common at their border, so that $|\pi|\leq|\rho(\pi)|-1$. 
If $\rho_3$ is adjacent to $\{\rho_1,\rho_2\}$, they also have a $j$-value common at their border, so that also $|\pi|\leq|\rho(\pi)|-2$. 
We can continue in this way for $\rho_4,\rho_5,...,\rho_r$, and we obtain in the end that $|\pi|\leq|\rho(\pi)|-r(\pi)+1$. 
It is also clear from this construction, by considering different border possibilities for the $\rho_1,\rho_2,...$, that 
both equality and strict inequality can occur.
\end{proof}

In the following we will denote the set of partitions where (\ref{equalityoccurs}) holds by ${\cal B}(n)$ 
(note that ${\cal B}(n)$ will also depend on $\sigma_1$, but this dependency will be implicitly assumed, and will thus not be mentioned in the following). 
Writing $\rho(\pi)=\{ W_1,...,W_{|\rho(\pi)|}\}$, there are $|\rho(\pi)|$ independent phases in the corresponding term, 
which we denote $\omega_{W_1},...,\omega_{W_{|\rho(\pi)|}}$.
Write $W_j=W_j^{\cdot}\cup W_j^H$, 
where $W_j^{\cdot}$ consists of the even elements of $W_j$ (corresponding to the ${\bf V}$-terms), 
$W_j^H$ consists of the odd elements of $W_j$ (corresponding to the ${\bf V}^H$-terms). 
(\ref{moregeneral2}) can now be written (computations are similar to Appendix 1 in~\cite{ryandebbah:vandermonde1})
\begin{eqnarray}
  & & \sum_{\pi\in{\cal P}(n)} \sum_{(i_1,...,i_n)} \sum_{ \stackrel{ (j_1,...,j_n) }{ \pi(j_1,...,j_n) = \pi} } \nonumber\\
  & & \hspace{1cm} \prod_{i=1}^{|\sigma_1|} N_i^{-|\sigma_{1i}|/2} L^{-1} \nonumber\\
  & & \hspace{1cm} \times \prod_{r=1}^{|\rho(\pi)|} \E \left( e^{j \left( \sum_{k\in W_r^H} i_{(k+1)/2+1} - \sum_{k\in W_r^{\cdot}} i_{k/2+1} \right)  \omega_{W_r}} \right) \nonumber\\
  & & \hspace{1cm} \times {\bf D}_1(N)(j_1,j_1) \times\cdots\times {\bf D}_n(N)(j_n,j_n) \label{limitmoment3}.
\end{eqnarray}
Since $\E\left(e^{jn\omega}\right)=0$ when $\omega$ is uniform and $n\neq 0$, we get that the $i_1,...,i_n$ contribute in (\ref{limitmoment3}) only if 
\begin{equation} \label{solvethese}
  \sum_{k\in W_r^H} i_{(k+1)/2+1} = \sum_{k\in W_r^{\cdot}} i_{k/2+1}
\end{equation}
for $1\leq r\leq |\rho(\pi)|$. 
The coefficient matrix of this system, denoted $A$, is a $|\rho(\pi)|\times n$ with entries from $\{-1,0,1\}$. 
The rank of $A$ is at most $k-1$, since the sum of all rows is $0$. 
Note that the number of solutions to (\ref{solvethese}) can also be written
\begin{equation} \label{kpundefmot}
  \int_{[0,2\pi)^{|\rho(\pi)|}} F(x) dx_1\cdots dx_{|\rho(\pi)|},
\end{equation}
where
\begin{equation} \label{fudef}
  F(x) = \prod_{k=1}^n \frac{1-e^{j N_{i_{2k}} (x_{b(2k-1)} - x_{b(2k)})}}{1-e^{j (x_{b(2k-1)} - x_{b(2k)})}},
\end{equation}
where $b(k)$ means the block in $\rho(\pi)$ which $k$ belongs to. 
This follows (as in~\cite{ryandebbah:vandermonde1}) from summing over all possible choices of $i_1,...,i_n$ in (\ref{limitmoment3}), 
and using the formula for the sum of a geometric series. 

It is easily seen that the rank of $A$ is exactly $k-1$ when $\rho(\pi)\vee [0,1]_n=1_{2n}$. 
More generally, if $\rho(\pi)\vee [0,1]_n=\{ \rho_1,...,\rho_r\}$ with each $\rho_i \geq [0,1]_{\|\rho_i\|/2}$, 
the rank of the system is $|\rho(\pi)|-r$. 
This follows since the equations corresponding to each $\rho_i$ have no variables in common with those from other $\rho_j$, 
and since the sum of the equations corresponding to $\rho_i$ is $0$, so that the coefficient matrix corresponding to  $\rho_i$ 
has one less than full rank. 
Also, note that $N_{i_{s+2}}=N_{i_{t+2}}$ whenever $2s+1,2s+2,2t+1,2t+2$ all belong to the same such $\rho_i$, 
and denote this common value by $N_{\rho_i}$
(meaning that there is a common upper limit $N_{\rho_i}$ to all variables $i_l$ occurring in connection with the same block $\rho_i$). 
This means that the number of solutions to (\ref{solvethese}) is of order 
\begin{eqnarray*}
  \lefteqn{ O\left(\prod_{i=1}^r N_{\rho_i}^{\|\rho_i\|/2-|\rho_i|+1}\right) } \\
  &=& O\left(\prod_{i=1}^r L^{\|\rho_i\|/2-|\rho_i|+1}\right)=O\left(L^{n-|\rho|+r}\right).
\end{eqnarray*}
Since $r$ depends only on $\pi$, we will also write $r=r(\pi)$.
Since the number of solutions to (\ref{solvethese}) is given by (\ref{kpundefmot}), the limit 
\begin{equation} \label{kpundef}
  \lim_{L\rightarrow\infty} \frac{1}{\prod_{i=1}^r N_{\rho_i}^{\|\rho_i\|/2-|\rho_i|+1}} \int_{[0,2\pi)^{|\rho(\pi)|}} F(x) dx_1\cdots dx_{|\rho(\pi)|}
\end{equation}
exists (here $u$ denotes the uniform distribution), and we will denote this limit by $K_{\rho(\pi),u}$.
Moreover $K_{\rho(\pi),u}=\prod_{i=1}^r K_{\rho_i,u}$, 
since the splitting $\rho(\pi)\vee [0,1]_n=\{ \rho_1,...,\rho_r\}$ actually splits the equations into $r$ sets where each set has no variables in common 
with other sets. 
This definition extends that of Vandermonde mixed moment expansion coefficients from~\cite{ryandebbah:vandermonde1}
to the case where the equations (\ref{solvethese}) may have an unequal number of variables on each side.
Note that in~\cite{ryandebbah:vandermonde1}, those coefficients were defined in terms of $\pi\in{\cal P}(n)$, 
while here they are defined in terms of $\rho(\pi)$, which captures any $\sigma_1$, which is new in the analysis given here.
Also in accordance with~\cite{ryandebbah:vandermonde1},
we will denote by $K_{\rho(\pi),u,L}$ the quantity inside the limit of (\ref{kpundef}), so that the 
number of solutions to (\ref{solvethese}) is 
\begin{equation} \label{numsols}
\prod_{i=1}^r N_{\rho_i}^{\|\rho_i\|/2-|\rho_i|+1}K_{\rho(\pi),u,L}.
\end{equation}

The number of blocks in the partitions $\rho(\pi),\pi$ say how many distinct choices from $(i_1,...,i_n)$ and $(j_1,...,j_n)$, respectively, 
contribute in (\ref{limitmoment3}). 
By substituting
\[ \sum_{j_k}{\bf D}_i(N)(j_k,j_k) = L\mathrm{tr}({\bf D}(N)), \]
and using (\ref{equalityoccurs}), we see that (\ref{limitmoment3}) is 
\begin{eqnarray*}
  \lefteqn{O(L^{-n-1+n-|\rho(\pi)|+r+|\pi|})} \\
  &\leq& O(L^{-n-1+n-|\rho(\pi)|+r+|\rho(\pi)| - r + 1}) = O(1),
\end{eqnarray*}
with equality if and only if $\pi\in{\cal B}(n)$ by Lemma~\ref{pilemma}. 
To check if $\pi$ belongs to ${\cal B}(n)$, $\rho(\pi)$ needs to be computed, and it is checked if equality in (\ref{equalityoccurs}) holds.
If so, the corresponding equation system (\ref{solvethese}) is constructed, 
and solved using Fourier-Motzkin elimination. Adding contributions for all partitions, we obtain (\ref{moregeneral2}).

For $\pi\in{\cal B}(n)$, 
noting that we can write $\prod_{i=1}^{|\sigma_1|} N_i^{-|\sigma_{1i}|/2}= \prod_{i=1}^{r} N_{\rho_i}^{-\|\rho_i\|/2}$, 
and using (\ref{numsols}),
we can write the contribution from $\pi$ in (\ref{limitmoment3}) as 
\begin{eqnarray*}
  \lefteqn{ \prod_{i=1}^{r} N_{\rho_i}^{-\|\rho_i\|/2} L^{-1} L^{|\pi|} \prod_{i=1}^r N_{\rho_i}^{\|\rho_i\|/2-(|\rho_i|-1)} K_{\rho(\pi),u,L} D_{\pi} } \\
  &=& L^{|\rho(\pi)| - r(\pi)} \prod_{i=1}^r N_{\rho_i}^{-(|\rho_i|-1)} K_{\rho(\pi),u,L} D_{\pi} \\
  &=& \prod_{i=1}^r L^{|\rho_i| - 1} \prod_{i=1}^r N_{\rho_i}^{-(|\rho_i|-1)} K_{\rho(\pi),u,L} D_{\pi} \\
  &=& \prod_{i=1}^r \left(\frac{L}{N_{\rho_i}}\right)^{|\rho_i|-1} K_{\rho(\pi),u,L} D_{\pi}.
\end{eqnarray*} 
Thus, if $\omega$ is uniform, taking limits in (\ref{limitmoment3}) gives
\begin{eqnarray*}
  \lefteqn{ \sum_{\pi\in{\cal B}(n)} \prod_{i=1}^r \left(\frac{L}{N_{\rho_i}}\right)^{|\rho_i|-1} K_{\rho(\pi),u,L} D_{\pi} } \\
  &\rightarrow& \sum_{\pi\in{\cal B}(n)} \prod_{i=1}^r c_{\rho_i}^{|\rho_i|-1} K_{\rho(\pi),u} D_{\pi} \\
  &=& \sum_{\pi\in{\cal B}(n)} \prod_{i=1}^r \left( c_{\rho_i}^{|\rho_i|-1} K_{\rho_i,u}\right) D_{\pi},
\end{eqnarray*}
where we have substituted $c_{\rho_i}=\lim_{L\rightarrow\infty} \frac{L}{N_{\rho_i}}$. 
When the $\omega_i$ are not uniform, we can still in (\ref{limitmoment3}) sum over the different $i_1,...,i_n$ to factor out the term
\begin{equation} \label{kpundefmot2}
  \int_{[0,2\pi)^{|\rho(\pi)|}} F(\omega) d\omega_1\cdots d\omega_{|\rho(\pi)|},
\end{equation}
where $F$ is defined by (\ref{fudef}), 
and where the only difference from (\ref{kpundefmot}) is that the uniform distribution $u$ has been replaced with $\omega$. 
The analysis is otherwise the same as in the uniform case, the major issue being the existence of the limits 
\begin{equation} \label{kpundef2}
  \lim_{L\rightarrow\infty} \frac{1}{\prod_{i=1}^r N_{\rho_i}^{\|\rho_i\|/2-|\rho_i|+1}} \int_{[0,2\pi)^{|\rho(\pi)|}} F(\omega) d\omega_1\cdots d\omega_{|\rho(\pi)|},
\end{equation}
which thus also will be called Vandermonde mixed moment expansion coefficients, and denoted $K_{\rho(\pi),\omega}$.
As in the uniform case, note that $K_{\rho(\pi),\omega}=\prod_{i=1}^r K_{\rho_i,\omega}$, and if these limits exist, we get as in the uniform case a limit on the form
\[
  \sum_{\pi\in{\cal B}(n)} \prod_{i=1}^r \left(c_{\rho_i}^{|\rho_i|-1} K_{\rho_i,\omega}\right) D_{\pi}.
\]
In~\cite{ryandebbah:vandermonde1}, it was shown that the limits $K_{\pi,\omega}$ exist when $\omega$ has a continuous density. 
We will show that the same holds for $K_{\rho(\pi),\omega}$, and the proof of this will follow from the following lemma, 
which is a generalization of Lemma 2 in~\cite{ryandebbah:vandermonde1}:

\begin{lemma} \label{appendixlemma}
Let $\rho(\pi)\leq\sigma_1\in{\cal P}(2n)$ be any partition such that $\rho(\pi)\vee [0,1]_n=1_{2n}$.
For any $\epsilon > 0$,
\begin{equation}
  \lim_{N\rightarrow\infty} \frac{1}{N^{n+1-|\rho(\pi)|}} \int_{B_{\epsilon , k}} F(\omega) d\omega = 0,
\end{equation}
where
\begin{equation}
  B_{\epsilon , k} = \{ (\omega_1,...,\omega_{|\rho(\pi)|}) | | \omega_{b(2k-1)} -\omega_{b(2k)} | > \epsilon \},
\end{equation}
and where $b(k)$ denotes the block in $\rho(\pi)$ which $k$ belongs to.
\end{lemma}

\begin{proof}
  The condition $\rho(\pi)\vee [0,1]_n=1_{2n}$ implies that when $\omega\in B_{\epsilon,k}$, $\omega_i-\omega_j<2n\epsilon$ for all $i,j$, 
  which means that the definition of $B_{\epsilon , k}$ is similar to the definition of $B_{\epsilon , r}$ in Lemma 2 in Appendix H in~\cite{ryandebbah:vandermonde1}. 
  The proof otherwise follows the same lines as~\cite{ryandebbah:vandermonde1}. 
\end{proof}

Using Lemma~\ref{appendixlemma} repeatedly, we see also that 
\[
  \lim_{N\rightarrow\infty} \frac{1}{N^{n+r-|\rho(\pi)|}} \int_{\cup_i B_{\epsilon , k_i}} F(\omega) d\omega = 0
\]
Letting $N\rightarrow\infty$, we obtain as in Appendix H of~\cite{ryandebbah:vandermonde1} in the limit
\begin{eqnarray*}
  \lefteqn{ \lim_{N\rightarrow\infty} \frac{1}{N^{n+r-|\rho(\pi)|}} \int F(\omega) d\omega } \\
  &=& K_{\rho(\pi),u} \prod_{i=1}^r (2\pi)^{|\rho_i\cap\sigma_{1j}|-1} \int \prod_j p_{\omega_j}(x)^{|\rho_i\cap\sigma_{1j}|}dx \\
  &=& K_{\rho(\pi),u} \prod_{i=1}^r (2\pi)^{|\rho_i\cap\sigma_j|-1} \int \prod_j p_{\omega_j}(x)^{|\rho_i\cap\sigma_j|}dx,
\end{eqnarray*}
where $\sigma_j$ are the blocks of $\sigma$.
Things have now been reduced to the case of uniform phase distribution.
In summary, (\ref{limitmoment3}) can be written
\begin{equation} \label{expressthis}
  \sum_{\pi\in{\cal B}(n)} D_{\pi} \prod_{i=1}^r \left((2\pi c_{\rho_i})^{|\rho_i|-1} K_{\rho_i,u}\right) 
  \prod_{i=1}^r \int \prod_j p_{\omega_j}(x)^{|\rho_i\cap\sigma_j|} dx,
\end{equation}
This is the standard form which the implementation uses, where the output of Fourier-Motzkin elimination is substituted into $K_{\rho(\pi),u}$. 
In (\ref{expressthis}), we recognize the integrals $I_{k,\omega}$ in (\ref{ikdef}). 
We will therefore substitute $I_{k,\omega}$ in the following. 

The requirement from Theorem~\ref{deconv} that $\sigma\geq [0,1]_n$ 
(which happens whenever terms of the form ${\bf V}_{\omega_1}^H{\bf V}_{\omega_2}$ 
(with $\omega_1,\omega_2$ different and ${\bf V}_{\omega_1},{\bf V}_{\omega_2}$ independent) 
do not occur) translates to the fact that, for any $\pi\in{\cal P}(n)$,  
in all $\rho_i$ the corresponding random matrices have equal phase distributions 
(with no assumptions on whether the random matrices are independent or not). 
From this it follows from (\ref{expressthis}) that
no integrand in (\ref{expressthis}) will contain two different densities. 
Therefore, the mixed moment is completely determined from the integrals $I_{k,\omega}$. 
To make the connection between these quantities and the moments, we need the following lemma, compiled from~\cite{ryandebbah:vandermonde1}:
\begin{lemma} \label{momentstoi}
Let ${\bf V}$ be a Vandermonde matrix with phase distribution $\omega$ and aspect ratio $c$.
For each $n$ there exists an invertible $n\times n$ matrix $A_n$ so that
\begin{equation} \label{insertthis}
  \left( \begin{array}{c} 1 \\ cI_{2,\omega} \\ c^2I_{3,\omega} \\ \vdots \\ c^{n-1}I_{n,\omega} \end{array} \right) 
  =
  A_n
  \left( \begin{array}{c} 1 \\ V_2 \\ V_3 \\ \vdots \\ V_n \end{array} \right)
\end{equation}
\end{lemma}

Inserting (\ref{insertthis}) in (\ref{expressthis}) when each $\rho_i$ consists of equal phase distributions, 
we obtain that (\ref{moregeneral2}) is completely determined from the moments $V^{(i)}_n$. 
Since there are $r$ integrals multiplied together in (\ref{expressthis}), its general form is seen to coincide with that of (\ref{mostgeneral}). 
We have thus proved Theorem~\ref{deconv}.
When terms of the form ${\bf V}_{\omega_1}^H{\bf V}_{\omega_2}$ occur, Section~\ref{counterex} shows that we can't expect dependence on only the moments. 
Instead the mixed moment depends on the entire phase distribution. 

\subsection{Handling the aspect ratio} \label{deconvproofc}
We will finally comment on appropriate forms of (\ref{expressthis}) which are useful in implementations of convolution.
We first turn to the case when there are deterministic matrices present. 
Assume that all matrix aspect ratios are equal to $c$, so that 
$\prod_{i=1}^r \left(c_{\rho_i}^{|\rho_i|-1}\right)=c^{|\rho|-r}=c^{|\pi|-1}$ when $\pi\in{\cal B}(n)$. 
Defining $m_n=cM_n$ and $d_n=cD_n$ as in~\cite{ryandebbah:vandermonde1}) in this case, (\ref{expressthis}) can also be written
\begin{equation}
  m_n = \sum_{\stackrel{\rho(\pi)\leq\sigma_1}{\pi\in{\cal B}(n)}} d_{\pi} K_{\rho(\pi),u} \prod_{i=1}^r \int \prod_j p_{\omega_j}(x)^{|\rho_i\cap\sigma_j|} dx,
\end{equation}
i.e. the aspect ratio $c$ can be handled as in~\cite{ryandebbah:vandermonde1}, providing a clear parallel with Proposition 3 in that paper. 

When there are no deterministic matrices present, and $\sigma\geq [0,1]_n$, the right hand side in (\ref{expressthis}) is
\begin{equation} \label{computethesefirst}
  \prod_{i=1}^r c_{\rho_i}^{|\rho_i|-1} I_{|\rho_i|,\omega_i} K_{\rho_i,u},
\end{equation} 
where we recognize the elements in the vector on the left hand side in (\ref{insertthis}). 
Therefore, an implementation of convolution would first compute the $c^{k-1}I_{k,\omega_i}$ using Lemma~\ref{momentstoi}, 
and substitute these directly into (\ref{computethesefirst}). 
For deconvolution, (\ref{computethesefirst}) would be computed first for one of the unknown component matrices, 
and then the moments would be recovered in a second step using Lemma~\ref{momentstoi}.

In summary, in order to compute the mixed moments (\ref{moregeneral2}) of Vandermonde matrices, we need to  
\begin{enumerate}
  \item iterate through partitions $\pi\in{\cal P}(n)$, compute $\rho(\pi)$, and determine whether $\pi\in{\cal B}(n)$,
  \item perform Fourier-Motzkin elimination in order to solve the set of equations given by (\ref{solvethese}), 
  \item compute the quantities in (\ref{computethesefirst}), either from direct knowledge of the phase distribution, or by 
    computing them from the moments using (\ref{insertthis}),
  \item compute the final result by inserting the results from 1), 2), and 3) into (\ref{expressthis}).
\end{enumerate}
This explains why Section~\ref{theoremimpl} focuses on the implementation perspectives of these tasks.

\section{The proofs of the convolution formulas} \label{appendixteo012}
In this appendix, we provide additional remarks, which together with the proof in Appendix~\ref{deconvproof}
will suffice to prove the different convolution formulas. 
We first provide a short explanation why convolution 2) does not depend only on the spectra of the component matrices. 
When ${\bf D}_i(N)$ only occurs in patterns of the form ${\bf V}{\bf D}(N){\bf V}^H$, we factored out 
the moments of ${\bf D}(N)$ in (\ref{limitmoment3}) in Appendix~\ref{deconvproof}. 
For other patterns, one ends instead up with integral expressions along the diagonal of ${\bf D}(N)$ 
(the diagonal elements of ${\bf D}(N)$ are multiplied with different complex exponentials), 
which are hard to express in terms of the moments of ${\bf D}(N)$. 

\subsection{The proof of Theorem~\ref{teo0}}
We can sum over all $\pi$ in (\ref{expressthis}) for these convolutions, 
since $r(\pi)=1$ and $|\rho(\pi)|=|\pi|$ for all $\pi$ whenever $\sigma=\sigma_1=1_{2n}$.
The implementation has obtained the result by inserting (\ref{insertthis}) into (\ref{expressthis}). 
$D_{\pi}$ in (\ref{expressthis}) can be handled in the following way:
Let $R_n$ be the set of multi-indices $r=(r_1,...,r_s)$ such that
\begin{itemize}
  \item The $r_i$ are decreasing, and all are integers $>0$.
  \item $\sum r_i=n$,
\end{itemize}
and set $d_r=\prod_{i=1}^s d_{r_i}$ for $r=(r_1,...,r_s)\in R_n$. Set also
\begin{eqnarray*}
  D_1 &=& (d_1) \\
  D_2 &=& (d_2,d_1^2) \\
  D_3 &=& (d_3,d_2d_1,d_1^3) \\
  D_4 &=& (d_4,d_2^2,d_3d_1,d_2d_1^2,d_1^4) \\
  D_5 &=& (d_5,d_3d_2,d_4d_1,d_2^2d_1,d_3d_1^2,d_2d_1^3,d_1^5), 
\end{eqnarray*}
and so on. It is clear from Appendix~\ref{deconvproof} that we can find a vector $K_n$ such that
\begin{equation} \label{togen}
  \left( \begin{array}{c} M_1 \\ M_2 \\ M_3 \\ \vdots \\ M_n \end{array} \right)
  =
  D_n^T K_n A_n \left( \begin{array}{c} 1 \\ V_2 \\ V_3 \\ \vdots \\ V_n \end{array} \right).
\end{equation}
Moreover, the matrices $K_n$ and $A_n$ can be computed once and for all. 

We see that there is only one term on the right hand side in (\ref{togen}) here containing $d_n$, 
so that this term can be found once $d_1,...,d_{n-1}$ have been found. 
This enables us to perform deconvolution.

\subsection{The proof of Theorem~\ref{teo1}}
Write $\mathrm{tr}\left( ({\bf D}+{\bf V}^H{\bf V})^n \right)$ as
\begin{eqnarray*}
 &=& \sum_{k,s} 
     \sum_{\begin{array}{c} \scriptsize (r_1,...,r_k) \\ \scriptsize \sum r_i=n-k-s \end{array}} \\
 & & \mathrm{tr}\left( \underbrace{{\bf D}\cdots {\bf D}}_{r_1\mbox{ times}} {\bf V}^H{\bf V} \underbrace{{\bf D}\cdots {\bf D}}_{r_2\mbox{ times}} {\bf V}^H{\bf V} \underbrace{{\bf D}\cdots {\bf D}}_{s\mbox{ times}}\right) \\
 &=& \sum_{k,s} \sum_{\begin{array}{c} \scriptsize (r_1,...,r_k) \\ \scriptsize \sum r_i=n-k \\ \scriptsize r_1\geq s \end{array}}
     \mathrm{tr}\left( {\bf D}^{r_1} {\bf V}^H{\bf V} \cdots {\bf D}^{r_k} {\bf V}^H{\bf V} \right).   
\end{eqnarray*}
Each summand here can be computed by inserting (\ref{insertthis}) into (\ref{expressthis}) as above.
Also, the multi-indices $(r_1,...,r_k)$ are easily traversed. 
It is clear that one can generalize (\ref{togen}) to compute each summand (the vector $K_n$ is simply expanded to handle more mixed moments). 
This explains how the implementation computes the formulas for Theorem~\ref{teo1}. 

Deconvolution for Theorem~\ref{teo1} follows the exact same argument as for Theorem~\ref{teo0}.

\subsection{The proof of Theorem~\ref{teo4}}
(\ref{multeq1}) corresponds to the case where 
\[
  \sigma=\sigma_1=\{\{ 1,2,5,6,9,10,...\},\{ 3,4,7,8,11,12,...\}\}.
\]
Following the notation in (\ref{limitmoment3}) in Appendix~\ref{deconvproof}, 
when $\pi=1_{2n}$ in (\ref{multeq1}), $\rho(\pi)=[0,1]_{2n}$, so that the contribution is
\[
  k \int p_{\omega_1}^n(x)dx  \int p_{\omega_2}^n(x)dx
\]
contributes, where $k$ is a scalar. Also, for all other choices of $\pi$, the integral $I_{n,\omega_2}$ does not contribute, 
so that the equation for th $n$'th moment uniquely determines $I_{n,\omega_2}$, when the lower order integrals $\{I_{k,\omega_2}\}_{k<n}$ are known. 
Due to Lemma~\ref{momentstoi}, the same can be said for the moments, so that it is possible to perform deconvolution. 

Similarly, the contribution from $\pi=1_{n}$ in (\ref{multeq2}) for the term when the second summand is always chosen is $k I_{n,\omega_2}$,
where $k$ is a scalar. Moreover, $I_{n,\omega_2}$ contributes only for this term and this $\pi$, so that 
the equation for th $n$'th moment uniquely determines $I_{n,\omega_2}$. 
It follows as above that it is possible to perform deconvolution.

\subsection{The proof of Theorem~\ref{teo2}}
This case corresponds to $\sigma=1_{2n}$, and 
\[\sigma_1=\{[2,3],[4,5],...,[2n-2,2n-1],[2n,1]\}.\]
From (\ref{expressthis}) it is clear that we can write
\begin{equation} \label{starthere2}
  \left( \begin{array}{c} M_1 \\ M_2 \\ \vdots \\ M_n \end{array} \right) = K_n \left( \begin{array}{c} V_1 \\ V_2 \\ \vdots \\ V_{2n} \end{array} \right),
\end{equation}
where $K_n$ is an $n\times 2n$ matrix, depending only on the values computed from Fourier-Motzkin elimination. 

Deconvolution in general for (Theorem~\ref{teo2}) is impossible, since the equation system (\ref{starthere2}) has twice as many unknowns as equations.
So, in this case, we need some prior knowledge about the phase distribution in order to perform deconvolution.

\subsection{The proofs of Theorem~\ref{teo3} and Theorem~\ref{teohankel}}
For Toeplitz matrices,~\cite{paper:brycdembojiang} shows that we can compute the moments in the same way as for Vandermonde matrices, 
but that we need only consider equations on the form (\ref{volumeeq}) with all blocks of $\rho$ of cardinality two.
The case of Hankel matrices is similar, however here the variables in (\ref{volumeeq}) 
are placed differently on the left and right sides\footnote{In the software described for this paper, Toeplitz, Hankel, and Vandermonde matrices all reuse the same code, 
but different sets of partitions are considered, depending on the type of the matrix. 
Also, the way the corresponding equation is constructed from the partition depends on the type of the matrix}.

\section{The proof of Theorem~\ref{teonew}} \label{appendixteonew}
Assume first that all aspect ratios are equal to $c$. 
In (25) in Theorem 7 in~\cite{ryandebbah:vandermonde1}, set ${\bf D}_i(N)=I_L$, and place the last matrix ${\bf V}_{i_1}$ in front instead to obtain
  \begin{equation} \label{thisgen}
  \begin{array}{l}
    \lim_{N\rightarrow\infty} \E[ \mathrm{tr} ( {\bf V}_{i_1} {\bf V}_{i_1}^H {\bf V}_{i_2} {\bf V}_{i_2}^H \times \cdots \times {\bf V}_{i_n} {\bf V}_{i_n}^H  ) ] \\
    =
    \sum_{\rho\leq\sigma\in{\cal P}(n)} K_{\rho , \omega} c^{|\rho |}
  \end{array}
  \end{equation}
(note that $c^{|\rho |}$ appears instead of $c^{|\rho |-1}$ since ${\bf V}_{i_1}$ is moved to front, and thus the additional $c$-factor is due to the fact that 
we take the trace of a matrix with different dimensions). 
As in Appendix~\ref{deconvproof} it is straightforward to generalize this to the case where the independent Vandermonde matrices have different aspect ratios, 
i.e. (\ref{thisgen}) is
\[
  \sum_{\rho\leq\sigma\in{\cal P}(n)} K_{\rho , \omega} \prod_{i=1}^{|\sigma|}c_i^{|\rho\cap\sigma_i|},
\]
where $\rho\cap\sigma_i$ is the partition consisting of the blocks of $\rho$ contained in $\sigma_i$.
Using Theorem 8 in~\cite{ryandebbah:vandermonde1} (i.e. we also assume that the phase distributions are different, with ${\bf V}_i$ having phase distribution $\omega_i$), we thus generalize (25) to
\begin{eqnarray*}
  & & \sum_{\sigma\geq\rho} K_{\rho , u} (2\pi)^{|\rho |-1} \int_0^{2\pi} \prod_{i=1}^{s}  p_{\omega_i}(x)^{|\rho\cap\sigma_i|} dx \prod_{i=1}^{|\sigma|}c_i^{|\rho\cap\sigma_i|} \\
  &=& \sum_{\sigma\geq\rho} K_{\rho , u} (2\pi)^{|\rho |-1} \int_0^{2\pi} \prod_{i=1}^{s}  (c_i p_{\omega_i}(x))^{|\rho\cap\sigma_i|} dx \\
  &=& K_{\rho , u} (2\pi)^{|\rho |-1} \int_0^{2\pi} \prod_{i=1}^{s}  (c_1 p_{\omega_1}(x) + c_2 p_{\omega_1}(x))^{|\rho|} dx \\
  &=& (c_1+c_2)^{|\rho|} K_{\rho , u} (2\pi)^{|\rho |-1} \\
  & & \times \int_0^{2\pi} \prod_{i=1}^{s}  \left(\frac{1}{c_1+c_2}(c_1 p_{\omega_1}(x) + c_2 p_{\omega_1}(x))\right)^{|\rho|} dx \\
  &=& \lim_{N\rightarrow\infty} \E[ \mathrm{tr} ( {\bf V}_{\omega_1\ast_{c_1,c_2}\omega_2,c_1+c_2} {\bf V}_{\omega_1\ast_{c_1,c_2}\omega_2,c_1+c_2}^H )^n ],
\end{eqnarray*}
where we have used (\ref{thisgen}) on the density $\frac{1}{c_1+c_2}(c_1 p_{\omega_1}(x) + c_2 p_{\omega_1}(x))$.

\section{The proof of Theorem~\ref{almostsuretheorem}} \label{appendixteoas}
We will first concentrate on the proof for almost sure convergence for a single Vandermonde matrix, as this case is the simplest. 
This proof will follow the same lines as that of almost sure convergence in~\cite{paper:brycdembojiang}, 
in that one uses Chebyshev's inequality, the Borel Cantelli lemma, and the following result:
\begin{lemma} \label{firstlemma}
Assume that ${\bf V}$ is an ensemble of random Vandermonde matrices with a continuous phase distribution, such that $\frac{L}{N}\rightarrow c$.
For any $r\geq 1$ there exists a constant $C_r$ such that, for all $L$, 
\begin{equation} \label{provethis}
  \E\left[ \left( \mathrm{tr}\left( \left( {\bf V}^H{\bf V} \right)^r \right)  - \E\left[ \mathrm{tr}\left( \left( {\bf V}^H{\bf V} \right)^r \right) \right] \right)^4 \right] \leq C_rL^{-3}.
\end{equation}
\end{lemma}

Comparing with~\cite{paper:brycdembojiang,book:hiaipetz}, Lemma~\ref{firstlemma} suggests that Vandermonde matrices 
converge somewhat faster than Hankel- and Toeplitz matrices, 
but somewhat slower than Gaussian matrices. 

\begin{proof}
We can write
\begin{eqnarray} 
\lefteqn{\E\left[ \left( \mathrm{tr}\left( \left( {\bf V}^H{\bf V} \right)^r \right)  - \E\left[ \mathrm{tr}\left( \left( {\bf V}^H{\bf V} \right)^r \right) \right] \right)^4 \right]} \nonumber \\
&=&     \E\left[ \left( \mathrm{tr}\left( \left( {\bf V}^H{\bf V} \right)^r \right) \right)^4 \right] \nonumber \\
& & - 4 \E\left[ \mathrm{tr}\left( \left( {\bf V}^H{\bf V} \right)^r \right) \right]\E\left[ \left( \mathrm{tr}\left( \left( {\bf V}^H{\bf V} \right)^r \right) \right)^3 \right] \nonumber \\
& & + 6 \left( \E\left[ \mathrm{tr}\left( \left( {\bf V}^H{\bf V} \right)^r \right) \right]\right)^2 \E\left[ \left( \mathrm{tr}\left( \left( {\bf V}^H{\bf V} \right)^r \right) \right)^2 \right] \nonumber \\
& & - 3 \left( \E\left[ \mathrm{tr}\left( \left( {\bf V}^H{\bf V} \right)^r \right) \right] \right)^4. \label{neweq3} 
\end{eqnarray}
We use certain interval partitions to define the following classes of partitions in ${\cal P}(4r)$:
\begin{itemize}
  \item ${\cal P}_0$: partitions $\pi$ such that
    \[\pi\leq\{[1,r],[r+1,2r],[2r+1,3r],[3r+1,4r]\},\]
  \item ${\cal P}_{1,2}$: partitions $\pi\not\in{\cal P}_0$ such that
    \[\pi\leq\{[1,2r],[2r+1,3r],[3r+1,4r]\},\] 
  \item ${\cal P}_{2,3}$: partitions $\pi\not\in{\cal P}_0$ such that
    \[\pi\leq\{[1,r],[r+1,3r],[3r+1,4r]\}\]
    (all other ${\cal P}_{i,j}$ are defined similarly),
  \item ${\cal P}_{1,2,3}$: partitions 
    $\pi\not\in {\cal P}_0\cup{\cal P}_{1,2}\cup{\cal P}_{1,3}\cup{\cal P}_{1,4}\cup{\cal P}_{2,3}\cup{\cal P}_{2,4}\cup{\cal P}_{3,4}$
    such that 
    \[\pi\leq\{[1,3r],[3r+1,4r]\}\] 
    (all other ${\cal P}_{i,j,k}$ are defined similarly),
  \item ${\cal P}_{1,2,3,4}$: partitions which are in none of the sets ${\cal P}_0,{\cal P}_{i,j},{\cal P}_{i,j,k}$.
\end{itemize}
These classes of partitions are indexed by which intervals in $\{[1,r],[r+1,2r],[2r+1,3r],[3r+1,4r]\}$ 
are joined (in the sense that at least one block in a partition $\pi$ in ${\cal P}_{1,2}$ 
should contain elements from both the first and second interval in $\{[1,r],[r+1,2r],[2r+1,3r],[3r+1,4r]\}$), 
and we can write ${\cal P}(4r)$ as a disjoint union:
\begin{eqnarray}
  {\cal P}(4r) &=& {\cal P}_0 \cup {\cal P}_{1,2} \cup {\cal P}_{1,3} \cup {\cal P}_{1,4} \cup {\cal P}_{2,3} \cup {\cal P}_{2,4} \cup {\cal P}_{3,4} \nonumber \\
               & & \cup {\cal P}_{1,2,3} \cup {\cal P}_{1,2,4} \cup {\cal P}_{1,3,4} \cup {\cal P}_{2,3,4} \nonumber \\
               & & \cup {\cal P}_{1,2,3,4}. \label{disjointpart}
\end{eqnarray}
We will denote the set of sets on the right hand side in (\ref{disjointpart}) by ${\cal S}$. 
Write 
\begin{eqnarray}
  S_{T,\pi}  &=& \sum_{ \stackrel{(j_1,...,j_{4r})}{\pi(j_1,...,j_{4r})=\pi}} \sum_{(i_1,...,i_{4r})} N^{-4r} L^{-4} \nonumber \\
             & & \hspace{1cm} \times \E_T \left( \prod_{k=1}^n \left( e^{j (\omega_{b(k-1)} -\omega_{b(k)}) i_k} \right) \right), \label{boundthis}
\end{eqnarray}
where 
\begin{enumerate}
  \item $\pi=\pi(j_1,...,j_{4r})$ is defined as in Appendix~\ref{deconvproof}, 
  \item $T$ is a subset of $\{1,2,3,4\}$,
  \item $\E_T(x_1\cdots x_4)=\E(\prod_{i\in T} x_i)\prod_{i\in T^c}\E(x_i)$ (i.e. $T$ dictates which random variables $x_i$ are grouped within the same expectation),
  \item $b(k)$ means the block of $\pi$ which $k$ belongs to (with $\omega_{W_1},...,\omega_{W_s}$ independent when $W_1,...W_s$ are the blocks of $\pi$).
  \item $k-1$ is formed modulo $\{[1,r],[r+1,2r],[2r+1,3r],[3r+1,4r]\}$, 
    meaning that the values of $k\rightarrow k-1$ actually takes the form
    \begin{eqnarray*}
      1    &\rightarrow& r \\
      r+1  &\rightarrow& 2r \\
      2r+1 &\rightarrow& 3r \\
      3r+1 &\rightarrow& 4r \\
      k    &\rightarrow& k-1 \mbox{, } k\not\in\{ 1,r+1,2r+1,3r+1\},
    \end{eqnarray*}
  \item $N^{-4r}$ are all normalizing factors in (\ref{vandermonde}), $L^{-4}$ are the normalizing factors which come from taking the four traces 
    for each term in (\ref{neweq3}).
\end{enumerate}
When we write out (\ref{neweq3}) (by writing out the matrix product as in Appendix~\ref{deconvproof}, 
we end up with sums of the form $\sum_{\pi} S_{T,\pi}$, with various values for $T$. 
We have in particular
\begin{eqnarray*}
  \sum_{\pi\in{\cal P}(4n)} S_{\{1,2,3,4\},\pi} &=& \E\left[ \left( \mathrm{tr}\left( \left( {\bf V}^H{\bf V} \right)^r \right) \right)^4 \right] \\
  \sum_{\pi\in{\cal P}(4n)} S_{\{\},\pi}        &=& \left( \E\left[ \mathrm{tr}\left( \left( {\bf V}^H{\bf V} \right)^r \right) \right] \right)^4.
\end{eqnarray*}
However, since only one Vandermonde matrix appears here, the analysis from Appendix~\ref{deconvproof} simplifies to the case $\sigma_1=\sigma=1_{4r}$, 
for which the quantities can be expressed directly in terms of $\pi\in{\cal P}(4r)$ rather than $\rho(\pi)\in{\cal  P}(8r)$ (as in Appendix~\ref{deconvproof}), 
so that the notation from~\cite{ryandebbah:vandermonde1}) can be followed more closely. 
As with (\ref{limitmoment3}), (\ref{neweq3}) thus becomes
\begin{eqnarray}
  & & \sum_{\pi\in{\cal P}(4r)} \left( S_{\{ 1,2,3,4\},\pi} - 4 S_{\{ 2,3,4\},\pi} + 6 S_{\{ 3,4\},\pi} - 3S_{\{\},\pi} \right) \nonumber \\
  &=& \sum_{S\in{\cal S}}\sum_{\pi\in S} (S_{\{ 1,2,3,4\},\pi} - 4S_{\{ 2,3,4\},\pi} \nonumber \\
  & & \hspace{1cm} + 6S_{\{ 3,4\},\pi} - 3S_{\{\},\pi}), \label{sumthese}
\end{eqnarray}
due to the ordering of the expectations in (\ref{neweq3}). 
We now consider all possibilities for $S\in{\cal S}$ in (\ref{sumthese}). 
For $\pi\in{\cal P}_0$ it is clear that one can split the expectations further to obtain 
\[
  S_{\{ 1,2,3,4\},\pi} = S_{\{ 2,3,4\},\pi} = S_{\{ 3,4\},\pi} = S_{\{\},\pi},
\]
and by adding up we see that the contribution from $S\in{\cal P}_0$ in (\ref{sumthese}) is $0$. 
Similarly, by splitting up the expectations as much as possible, the contributions for other $\pi$ in (\ref{sumthese}) is seen to be
\begin{eqnarray*}
  \pi\in{\cal P}_{1,2}     &:& S_{\{ 1,2\},\pi} - 4 S_{\{\},\pi} + 6 S_{\{\},\pi} - 3S_{\{\},\pi} \\
  \pi\in{\cal P}_{1,3}     &:& S_{\{ 1,3\},\pi} - 4 S_{\{\},\pi} + 6 S_{\{\},\pi} - 3S_{\{\},\pi} \\
  \pi\in{\cal P}_{1,4}     &:& S_{\{ 1,4\},\pi} - 4 S_{\{\},\pi} + 6 S_{\{\},\pi} - 3S_{\{\},\pi} \\
  \pi\in{\cal P}_{2,3}     &:& S_{\{ 2,3\},\pi} - 4 S_{\{ 2,3\},\pi} + 6 S_{\{\},\pi} - 3S_{\{\},\pi} \\
  \pi\in{\cal P}_{2,4}     &:& S_{\{ 2,4\},\pi} - 4 S_{\{ 2,4\},\pi} + 6 S_{\{\},\pi} - 3S_{\{\},\pi} \\
  \pi\in{\cal P}_{3,4}     &:& S_{\{ 3,4\},\pi} - 4 S_{\{ 3,4\},\pi} + 6 S_{\{3,4\},\pi} - 3S_{\{\},\pi} \\
  \pi\in{\cal P}_{1,2,3}   &:& S_{\{ 1,2,3\},\pi} - 4 S_{\{ 2,3\},\pi} + 6 S_{\{\},\pi} - 3S_{\{\},\pi} \\
  \pi\in{\cal P}_{1,2,4}   &:& S_{\{ 1,2,4\},\pi} - 4 S_{\{ 2,4\},\pi} + 6 S_{\{\},\pi} - 3S_{\{\},\pi} \\
  \pi\in{\cal P}_{1,3,4}   &:& S_{\{ 1,3,4\},\pi} - 4 S_{\{ 3,4\},\pi} + 6 S_{\{ 3,4\},\pi} - 3S_{\{\},\pi} \\
  \pi\in{\cal P}_{2,3,4}   &:& S_{\{ 2,3,4\},\pi} - 4 S_{\{ 2,3,4\},\pi} + 6 S_{\{ 3,4\},\pi} - 3S_{\{\},\pi} \\
  \pi\in{\cal P}_{1,2,3,4} &:& S_{\{ 1,2,3,4\},\pi} - 4 S_{\{ 2,3,4\},\pi} + 6 S_{\{ 3,4\},\pi} - 3S_{\{\},\pi}.
\end{eqnarray*}
Adding everything here, 
and using that the contributions from ${\cal P}_{i,j,k}$ all are equal for different $i,j,k$, 
and that the contributions from ${\cal P}_{i,j}$ all are equal for different $i,j$
(which is obvious by associating each interval $[kr+1,(k+1)r]$ with $r$ values on a circle, noting that 
the different classes of partitions can be viewed as different ways of connecting the circles, 
and that the actual circles being joined does not matter for the final value), we obtain that (\ref{neweq3}) equals
\[
  \sum_{\pi\in{\cal P}_{1,2,3,4}} \left( S_{\{ 1,2,3,4\},\pi} - 4 S_{\{ 2,3,4\},\pi} + 6 S_{\{ 3,4\},\pi} - 3S_{\{\},\pi} \right),
\]
i.e. we need only sum over $\pi\in{\cal P}_{1,2,3,4}$ (all other terms cancel).
If the phase distribution is uniform, we consider the coefficient matrix for the equation system corresponding to $\pi\in{\cal P}_{1,2,3,4}$ 
(formed as in Appendix~\ref{deconvproof}). 
This has rank $|\pi|-1$, so that the number of solutions $(i_1,...,i_{4r})$
solving the equation system has order $N^{4r-|\pi|+1}$. 
Since the number of $j_1,...,j_{4r}$ such that $\pi(j_1,...,j_{4r})=\pi$ is of order $O\left(L^{|\pi|}\right)$, (\ref{boundthis}) is 
\begin{equation} \label{orderestimate}
  O\left(N^{-4r} L^{-4}L^{|\pi|}N^{4r-|\pi|+1}\right) = O(L^{-3}).
\end{equation}
This proves the claim for the uniform distribution. 
When the Vandermonde matrices do not have uniform phase distribution, 
as long as the phase distribution is continuous, we can reduce to the case of uniform phase distribution using Lemma~\ref{appendixlemma} and the 
techniques in Appendix~\ref{deconvproof}. The constant $C_r$ needs only to be modified by taking into account the maximum of all 
Vandermonde mixed moment expansion coefficients of order $4r$. 
\end{proof}

To prove the general case, we must in (\ref{provethis}) replace ${\bf V}^H{\bf V}$ with the combination appearing in (\ref{moregeneral2}). 
One in this case instead considers the interval partition 
\[ \{[1,nr],[nr+1,2nr],[2nr+1,3nr],[3nr+1,4nr]\} \]
instead of the interval partition $\{[1,r],[r+1,2r],[2r+1,3r],[3r+1,4r]\}$. 
The sets of partitions ${\cal P}_0,{\cal P}_{1,2},...$ are defined similarly, and they are now sets in ${\cal P}(4rn)$.
For a mixed moment as in (\ref{neweq3}) (with ${\bf V}^H{\bf V}$ replaced with combinations as in(\ref{moregeneral2})), 
one shows as before that only partitions in ${\cal P}_{1,2,3,4}$ contribute (i.e. all other terms cancel as above). 
Since the form (\ref{moregeneral2}) is used, one needs to construct the partition $\rho(\pi)\in{\cal P}(8rn)$ from $\pi\in{\cal P}(4rn)$, 
and as in Appendix~\ref{deconvproof}, only partitions satisfying (\ref{equalityoccurs}) contribute (i.e. $\pi\in{\cal B}(4rn)$), and (\ref{orderestimate}) 
becomes in this case 
\begin{eqnarray*}
  \lefteqn{ O\left(N^{-4rn} L^{-4}L^{|\pi|}L^{4rn-|\rho(\pi)|+r(\pi)}\right) } \\
  &=& O\left(N^{-4rn} L^{-4}L^{|\pi|}L^{4rn + 1 -|\pi|}\right) = O(L^{-3}),
\end{eqnarray*}
and the result follows.

\bibliography{../bib/mybib,../bib/mainbib}

\end{document}